\documentclass[11pt,a4paper,reqno]{amsart}%

\usepackage[margin=35mm]{geometry}

\usepackage[utf8]{inputenc}
\usepackage[english]{babel}
\usepackage{graphicx}         
\usepackage[pdftex, colorlinks=true, linkcolor=blue, citecolor=blue, urlcolor=black]{hyperref}
\usepackage{braket}
\usepackage{amsmath}
\usepackage{amssymb}
\usepackage{amsthm}
\usepackage{mathtools}
\usepackage{bbm}		
\usepackage{mathrsfs}
\usepackage{array}		
\usepackage{color}
\usepackage[usenames,dvipsnames]{xcolor}
\usepackage{caption}
\usepackage{enumitem}
\usepackage{subeqnarray}	
\usepackage{wrapfig}

\usepackage{tikz}
\usetikzlibrary{arrows, decorations.markings}

\allowdisplaybreaks

\numberwithin{equation}{section}

\theoremstyle{plain}\newtheorem{definition}{Definition}
\newtheorem{theorem}[definition]{Theorem}
\newtheorem{lem}[definition]{Lemma}
\newtheorem{proposition}[definition]{Proposition}
\newtheorem{cor}[definition]{Corollary}
\newtheorem{assumption}[definition]{Assumption}
\theoremstyle{remark}\newtheorem{remark}[definition]{Remark}

\theoremstyle{plain}

\newcommand{\R}{\mathbb{R}}

\newcommand{\fH}{\mathfrak{H}}
\newcommand{\Fock}{\mathcal{F}}
\newcommand{\Number}{\mathcal{N}}

\newcommand{\id}{\mathbbm{1}}
\newcommand{\1}{\mathbbm{1}}

\newcommand{\cR}{\mathcal{R}}

\renewcommand{\i}{\mathrm{i}}

\newcommand{\hc}{\mathrm{h.c.}}

\newcommand{\sym}{\mathrm{sym}}

\newcommand{\Tr}{\mathrm{Tr}}

\renewcommand{\hat}[1]{\widehat{#1}}
\renewcommand{\tilde}[1]{\widetilde{#1}}

\newcommand{\lr}[1]{\left\langle #1 \right\rangle}
\newcommand{\norm}[1]{\lVert#1\rVert}

\renewcommand{\d}{\mathop{}\!\mathrm{d}}
\newcommand{\dx}{\d x}
\newcommand{\dy}{\d y}
\newcommand{\dz}{\d z}

\newcommand{\ds}{\d s}

\newcommand{\ad}{a^\dagger}

\newcommand\mydots{,\makebox[1em][c]{.\hfil.\hfil.},}

\makeatletter
\DeclareFontFamily{OMX}{MnSymbolE}{}
\DeclareSymbolFont{MnLargeSymbols}{OMX}{MnSymbolE}{m}{n}
\SetSymbolFont{MnLargeSymbols}{bold}{OMX}{MnSymbolE}{b}{n}
\DeclareFontShape{OMX}{MnSymbolE}{m}{n}{
    <-6>  MnSymbolE5
   <6-7>  MnSymbolE6
   <7-8>  MnSymbolE7
   <8-9>  MnSymbolE8
   <9-10> MnSymbolE9
  <10-12> MnSymbolE10
  <12->   MnSymbolE12
}{}
\DeclareFontShape{OMX}{MnSymbolE}{b}{n}{
    <-6>  MnSymbolE-Bold5
   <6-7>  MnSymbolE-Bold6
   <7-8>  MnSymbolE-Bold7
   <8-9>  MnSymbolE-Bold8
   <9-10> MnSymbolE-Bold9
  <10-12> MnSymbolE-Bold10
  <12->   MnSymbolE-Bold12
}{}
\let\llangle\@undefined
\let\rrangle\@undefined
\DeclareMathDelimiter{\llangle}{\mathopen}
                     {MnLargeSymbols}{'164}{MnLargeSymbols}{'164}
\DeclareMathDelimiter{\rrangle}{\mathclose}
                     {MnLargeSymbols}{'171}{MnLargeSymbols}{'171}
\makeatother

\newcommand{\PhiN}{\Phi_N}
\newcommand{\PhiNM}{\Phi_{N,M}}
\newcommand{\GN}{\mathcal{G}_N}
\newcommand{\wN}{w_N}

\newcommand{\fHp}{\fH_\perp}
\newcommand{\Fp}{\Fock_\perp}
\newcommand{\FNp}{\Fock^{\leq N}_\perp}
\newcommand{\FMp}{\Fock^{\leq M}_\perp}

\newcommand{\uN}{u_N}
\newcommand{\mN}{\mu_N}
\newcommand{\hN}{h}
\newcommand{\Tmax}{T_\mathrm{max}}
\newcommand{\tN}{\theta_N}

\newcommand{\bH}{\mathbb{H}}

\newcommand{\bG}{\mathbb{G}}

\newcommand{\bGo}{\bG_1}
\newcommand{\bGt}{\bG_2}
\newcommand{\bGth}{\bG_3}

\newcommand{\idN}{\id^{\leq N}}
\newcommand{\idM}{\id^{\leq M}}
\newcommand{\idm}{\id^{\leq m}}

\newcommand{\idgm}{\id^{> m}}

\newcommand{\fM}{f_M}

\newcommand{\cB}{\mathcal{B}}

\newcommand{\dGo}{\d\Gamma_1}
\newcommand{\dG}{\d\Gamma}
\newcommand{\dGt}{\d\Gamma_2}

\newcommand{\Chi}{\boldsymbol{\chi}}

\newcommand{\gNo}{\gamma_{\Psi_N}^{(1)}}
\newcommand{\gNoz}{\gamma_{\Psi_{N,0}}^{(1)}}
\newcommand{\gNot}{\gamma_{\Psi_N(t)}^{(1)}}

\newcommand{\eps}{\varepsilon}
\newcommand{\nn}{\nonumber}
\newcommand{\cE}{\mathcal{E}}
\newcommand{\cN}{\mathcal{N}}
\newcommand{\cF}{\mathcal{F}}
\newcommand{\cW}{\mathcal{W}}

\allowdisplaybreaks

\title[2D focusing dynamics in the instability regime]{Focusing dynamics of 2D  Bose gases in the instability regime}

\author[L.Boßmann]{Lea Boßmann} 
\address{Department of Mathematics, LMU Munich, Theresienstrasse 39, 80333 Munich, and Munich Center for Quantum Science and Technology, Schellingstr. 4, 80799 Munich, Germany} 
\email{bossmann@math.lmu.de}

\author[C. Dietze]{Charlotte Dietze}
\address{Department of Mathematics, LMU Munich, Theresienstrasse 39, 80333 Munich, Germany} 
\email{dietze@math.lmu.de}

\author[P.T. Nam]{Phan Th\`anh Nam}
\address{Department of Mathematics, LMU Munich, Theresienstrasse 39, 80333 Munich, Germany} 
\email{nam@math.lmu.de}

\begin{document}
\date{\today}
\maketitle

\begin{abstract}
\noindent
We consider the dynamics of a 2D Bose gas with an interaction potential of the form $N^{2\beta-1}w(N^\beta\cdot)$ for $\beta\in (0,3/2)$. The interaction may be chosen to be negative and large, leading to the instability regime where the corresponding focusing cubic nonlinear Schrödinger equation (NLS) may blow up in finite time. We show that to leading order, the $N$-body quantum dynamics can be effectively described by the NLS prior to the blow-up time. Moreover, we prove the validity of the Bogoliubov approximation, where the excitations from the condensate are captured in a norm approximation of the many-body dynamics. 
\end{abstract}

\section{Introduction}

Since the pioneering work of Bose and Einstein in 1924 \cite{Bose-24,Einstein-24}, and especially after the experimental realization of the Bose-Einstein condensation in 1995 \cite{Wieman-Cornell-95,Ketterle-95},  there has been a remarkable effort to understand the macroscopic behavior of interacting Bose gases from first  principles. From the mathematical point of view, the theory of interacting Bose gases goes back to Bogoliubov's 1947 paper \cite{Bogoliubov-47}, where he proposed an effective method to transform a weakly interacting Bose gas to a non-interacting one, subject to a modification of the kinetic operator due to the interaction effect. In the present work, we will focus on the rigorous derivation of Bose--Einstein condensation and Bogoliubov's theory for the dynamics of two dimensional bosonic systems, where large attractive interaction potentials admit blow-up phenomena. 

These blow-up phenomena have been observed in experiments with ulta-cold Bose gases \cite{Bradley-95,Cornish-00,Donley-01}. In these settings, first a repulsive interaction was used to prepare an initial state, and then the interaction was switched to attractive by means of Feshbach resonances. When the strength of the attractive interaction was increased beyond a critical threshold, a blow-up process happened, where a large fraction of the condensate was lost \cite{Roberts-01}. The goal of our analysis is to understand this behaviour from a mathematical point of view for a 2D system.

\subsection{Setting}
 In the framework of many-body quantum physics, the dynamics of a system of $N$ (spinless) bosons in $\R^2$ can be described by the linear $N$-body Schr\"odinger equation,
 \begin{equation}\label{eq:SE}
\begin{cases}
&\i\partial_t\Psi_N(t)=H_N\Psi_N(t),\\
&\Psi_N(0) = \Psi_{N,0}\,,
\end{cases}
\end{equation}
where the wave function $\Psi_N(t)$ belongs to $L_{s}(\R^{2N})$, the space of square integrable functions of $N$ variables in $\R^2$ satisfying the bosonic symmetry
\begin{equation}
\Psi_N(t,x_1, ..., x_N)= \Psi_N(t,x_{\sigma(1)}, ..., x_{\sigma(N)}) \quad \forall \sigma\in S_N, \;\forall x_i\in \R^2\,,
\end{equation}
where $S_N$ denotes the set of all permutations of $\{1\mydots N\}$. We will work on a non-relativistic system with short-range interactions, where the underlying Hamiltonian is typically given by 
\begin{equation}\label{HN}
H_N=\sum_{j=1}^N(-\Delta_j)+\frac{1}{N-1}\sum_{1\leq j<k\leq N} w_N(x_j-x_k)\,,
\end{equation}
where
\begin{equation}\label{wN}
\wN(x)=N^{2\beta}w(N^\beta x)\,,\qquad \beta>0
\end{equation}
with a real-valued, even and bounded potential $w$. We do not impose any positivity condition on $w$; in particular, the attractive case $w\le 0$ is allowed. 

When $w$ is bounded, the Hamiltonian $H_N$ is self-adjoint on $L_{s}(\R^{2N})$ with the same domain as the non-interacting Hamiltonian. Therefore, the linear Schr\"odinger equation \eqref{eq:SE} has a unique global solution $\Psi_N(t)=e^{-itH_N} \Psi_N(0)$ with $t\in \R$, for every initial state $\Psi_N(0) \in L_{s}(\R^{2N})$. The major challenge in the analysis  of \eqref{eq:SE} is that the relevant dimension grows fast as $N\to \infty$, making it very difficult to extract helpful information about the quantum system by numerical methods. Therefore, in practice, it is desirable to obtain collective descriptions by reasonable approximations, based on  suitable assumptions on the initial state. 

Roughly speaking, Bose--Einstein condensation (BEC) is the phenomenon where many particles occupy a common quantum state. In particular, this is the case when the $N$-body wave function is approximately given by a factorized state, namely
\begin{align} \label{eq:BEC-intro-0}
\Psi_N (t,x_1,x_2,...,x_N) \approx \varphi (t,x_1)\varphi(t,x_2)... \varphi(t,x_N)
\end{align}
in an appropriate sense. Here the normalized function $\varphi(t,\cdot)\in L^2(\R^2)$ describes the condensate and its evolution is governed by the cubic nonlinear Schr\"odinger equation (NLS)
\begin{equation}\label{NLS}
\begin{cases}
\i\partial_t\varphi(t,x)=\left(-\Delta_x+b|\varphi(t,x)|^2-\mu(t)\right)\varphi(t,x),\\
\varphi(0,x)=\varphi_0(x)\,,
\end{cases}
\end{equation}
where
\begin{equation}\label{b}
b=\int_{\R^2}w,\quad \mu(t)=\frac b 2\int_{\R^2} |\varphi(t,x)|^4\dx. 
\end{equation}
The equation \eqref{NLS} can be formally obtained from \eqref{eq:SE} using the assumption \eqref{eq:BEC-intro-0} and the fact that $w_N(x)=N^{2\beta} w(N^\beta x)\to b \delta(x)$ weakly. 

The coupling constant $b=\int w$ in \eqref{NLS} is crucial. The focusing case $b<0$ and the defocusing case $b>0$ correspond to rather different physical situations. In particular, we are interested in the focusing case where the NLS \eqref{NLS} may blow up in finite time, even if the initial datum $\varphi(0)$ is smooth \cite{Weinstein-83}.  The possibility of the finite-time blow up is closely related to instability, which we will explain below.

\subsection{Stability vs. Instability} Since the 2D cubic NLS  \eqref{NLS} is mass critical, it is well-known from the work of Weinstein \cite{Weinstein-83} that the possibility of the finite-time blow up for $H^1$-solution depends not only on the sign of the interaction, but also on its strength. To be precise, let us denote the critical interaction strength as the optimal constant $a^*>0$ in the Gagliardo--Nirenberg interpolation inequality\footnote{Equivalently, $a ^* = \norm{Q}_{L ^2 } ^2 $, where $Q$ is the unique positive radial solution of $-\Delta Q + Q - Q ^3 = 0.$}
\begin{equation}\label{eqn:b*}
\left(\int_{\R^2}|\nabla f(x)|^2 \d x \right)\left(\int_{\R^2}|f(x)|^2 \d x \right)\geq\frac{a^*}{2}\int_{\R^2}|f(x)|^4 \d x,  \qquad \forall\;f\in H^1(\R^2). 
\end{equation}
Then from \cite[Theorem 3.1 and Theorem 4.2]{Weinstein-83},  we have two distinct regimes: 
\begin{itemize}
\item \textbf{Stability regime}: if $b>-a^*$, then  \eqref{NLS} has a unique global solution for all initial data $\varphi_0\in H^1(\R^2)$ satisfying $\|\varphi_0\|_{L^2}=1$. 

\item \textbf{Instability regime}: if $b<- a^*$, then finite-time blow up occurs, for example, for any initial datum $\varphi_0\in H^1(\R^2)\cap L^2(\R^2;|x|^2 \d x)$ satisfying  $\|\varphi_0\|_{L^2}=1$ and
\begin{equation}
\int_{\R^2} |\nabla \varphi_0(x) |^2 \d x  + \frac{b}{2} \int_{\R^2} | \varphi_0(x) |^4 \d x <0\,.
\end{equation}
\end{itemize}

We also refer to Baillon--Cazenave--Figueira \cite{BaiCazFig-77} for an earlier result on the global existence for the 2D cubic NLS, Ginibre--Velo \cite{GinVel-79} for a more general existence theory, and Merle \cite{Merle-93} for a complete characterization of the minimal-mass blow-up solutions in the special case $b=-a^*$. For an analysis of the minimizer of the corresponding NLS energy functional as $b\to a^*$, see \cite{GuoSei-14}.
The existence of a universal blow-up profile was proved in \cite{MerRap-04}. A precise description of the blow-up solutions near the blow-up time was established in \cite{MerRap-05}. For works on the blow-up rate, we refer to \cite{MerRap-02,MerRap-03,Raphael-05,MerRap-06}.

For the $N$-body quantum dynamics \eqref{eq:SE}, the solution $\Psi_N(t)$ exists globally for every $L^2$-initial datum. Nevertheless, we can still discuss  stability and instability regimes by considering the boundedness of the energy per particle. 
\begin{itemize}
\item \textbf{Stability regime}: the system is stable of the second kind  if 
\begin{align} \label{eq:HN-CN}
H_N \ge -CN
\end{align}
for some constant $C>0$ independent of $N$ (see \cite{LieSei-10}).  
In principle, the many-body stability \eqref{eq:HN-CN} is stronger than the NLS stability. By testing \eqref{eq:HN-CN} against factorized states, it is not difficult to see that \eqref{eq:HN-CN} implies that $b=\int w \ge - a^*$. Conversely, if $\int w_- > -a^*$ for $w_-=\min\{w,0\}$ the negative part of $w$, then it is known that \eqref{eq:HN-CN} holds for $\beta\le 1/2$ \cite{Lewin-15} (see also \cite{LewNamRou-16,GuoLu-16,LewNamRou-17,LewNamRou-18,NamRou-20} for related bounds for trapped systems).
\item \textbf{Instability regime}: if $\int w < -a^*$, then we only have  (using $\|w_N\|_{L^\infty} \le CN^{2\beta}$)
\begin{align} \label{eq:HN-CN-beta}
H_N \ge -CN^{1+2\beta},
\end{align}
and the optimality of the lower bound can be seen by testing against factorized states. In particular, \eqref{eq:HN-CN-beta} allows the energy per particle to diverge to $-\infty$ as $N\to \infty$, which is consistent with  blow up  of the NLS \eqref{NLS}. 
\end{itemize}

\subsection{The derivation of NLS from many-body dynamics.} The rigorous derivation of the NLS has been studied since the 1970s, initiated by Hepp \cite{Hepp-74}, Ginibre--Velo \cite{GinVel-79b} and Spohn \cite{Spohn-80}. In the defocusing case $w\ge 0$, we refer to \cite{AdaGolTet-07} for 1D case, \cite{KirSchSta-11,JebLeoPic-16} for 2D,  \cite{ErdSchYau-07,ErdSchYau-09,ErdSchYau-10,Pickl-15} for 3D, 
\cite{CheHol-15,Bossmann-20} for the effectively 2D dynamics of strongly confined 3D systems, 
and the book \cite{BenPorSch-15} for further results. 

In the focusing case ($w\le 0$) in 2D, most of the existing works in the literature are based on the stability condition $\int |w_-| < a^*$. In this case, the focusing NLS \eqref{NLS} is globally well-posed, and its derivation from the many-body equation \eqref{eq:SE} was given by Chen--Holmer  \cite{CheHol-15} and  Jeblick--Pickl \cite{JebPic-18} under the technical addition of a trapping potential like  $V(x)=|x|^s$, enabling them to use stability of the second kind for $0<\beta<(s+1)/(s+2)$ by \cite{LewNamRou-17}. Since the stability \eqref{eq:HN-CN} was later extended to trapped systems for $0<\beta<1$ \cite{NamRou-20}, the approaches in   \cite{CheHol-15,JebPic-18} are conceptually applicable for that range of $\beta$. In another approach, Nam-Napi\'orkowski \cite{NamNap-19} used only a weaker form of \eqref{eq:HN-CN} 
(but they still require the stability condition $\int |w_-| < a^*$), Thus,  being able to removing the trapping potential for all $0<\beta<1$. 

In the present paper, we will give a novel derivation of the focusing NLS \eqref{NLS} which covers arbitrarily negative potentials $w$ and all $\beta\in (0,3/2)$. Without the stability condition $\int |w_-| < a^*$, one only has the very weak bound \eqref{eq:HN-CN-beta} instead of \eqref{eq:HN-CN}, and the methods in \cite{CheHol-15,JebPic-18,NamNap-19} do not apply, or apply only to a relatively small range of $\beta$. To our knowledge, for an arbitrarily negative potential, the derivation of the NLS \eqref{NLS} prior to the blow-up time is only available for $\beta<1/2$, following the methods in \cite{Pickl-10,NamNap-17a,NamNap-19,Chong-21}. Our extended range of $\beta$ is remarkably large, allowing to make connections to the typical  physical setting of dilute Bose gases beyond the mean-field regime, which requires at least $\beta>1/2$. Even in the stability regime, our result is new since we allow $\beta\geq1$.  Actually, we will derive  \eqref{NLS} from a stronger result, namely a norm approximation of the many-body quantum dynamics which also describes the fluctuations around the condensate in the spirit of Bogoliubov's theory. That result requires further notation and explanation,  which we defer to the next section.

\section{Main results} \label{sec:main}

Recall that we consider the Schr\"odinger equation \eqref{eq:SE} with the Hamiltonian $H_N$ given in \eqref{HN}, where $\wN(x)=N^{2\beta}w(N^\beta x)$ as in \eqref{wN}. We will give rigorous descriptions of the macroscopic behavior of the many-body dynamics $\Psi_N(t)=e^{-\i tH_N} \Psi_{N,0}$ when $N\to \infty$, including the NLS \eqref{NLS} as the leading order approximation, and a norm approximation in $L_s^{2}(\R^{2N})$ as the second order approximation.

We always impose the following condition on the interaction potential. 

\begin{assumption}\label{ass:w}
Let $w\in L^\infty(\R^2)$ be compactly supported and $w(x)=w(-x)\in\R$. 
\end{assumption}

Note that we do not put any assumption on the sign and the size of $w$. In particular, $w$ can be arbitrarily negative. 

\medskip

\subsection{Derivation of the NLS} Let us recall the following well-known result concerning the NLS \eqref{NLS} (see e.g. \cite[Theorem 4.10.1]{Cazenave}). 

\begin{lem}\label{lem:NLS} For every $b\in \mathbb{R}$ and $\varphi_0\in H^1(\R^2)$ with $\|\varphi_0\|_{L^2}=1$, there exists a unique solution $\varphi \in C([0,\Tmax),H^1(\R^2))$ of \eqref{NLS} with a unique maximal time $\Tmax \in (0,\infty]$. 
Moreover, if $\Tmax<\infty$, then
\begin{equation} \label{eq:blow-up}
\lim_{t \nearrow \Tmax } \|\varphi(t) \|_{H^1}= \infty.
\end{equation}
\end{lem}

For non-trivial interactions $w$, the many-body quantum state $\Psi_N(t)$ is not expected to be close to the factorized state $\varphi(t)^{\otimes N}$ in norm (see Theorem \ref{thm:2} below). Therefore, the leading order approximation \eqref{eq:BEC-intro-0} has to be understood in an average sense, which can be formulated properly in terms of reduced density matrices. For every normalized vector $\Psi_N \in L^2_s(\R^{2N})$, its one-body density matrix $\gNo$ is  a non-negative operator on $L^2(\R^2)$ with  kernel 
\begin{equation}
\gNo (x;y) = \int_{\R^{2(N-1)}} \Psi_N(x,x_2,...,x_N) \overline{\Psi_N(y,x_2,...,x_N)} \d x_2 ... \d x_N\,.
\end{equation}
Equivalently, it can be obtained by taking the partial trace 
\begin{equation}
\gNo = \Tr_{2\to N} |\Psi_N\rangle \langle \Psi_N|.
\end{equation}
Clearly, if $\Psi_N= \varphi^{\otimes N}$, then $\gNo = |\varphi \rangle \langle \varphi|$ (the rank-one projection onto $\varphi\in L^2(\R^2)$). 
In general, the approximation 
\begin{equation}
\gNo \approx |\varphi \rangle \langle \varphi|
\end{equation}
with respect to the trace norm is an appropriate interpretation of \eqref{eq:BEC-intro-0}.  
Our first main result is a rigorous derivation of the NLS \eqref{NLS} from \eqref{eq:SE}. 

\begin{theorem}[NLS evolution of the condensate]\label{thm:1} 
Let $\beta\in (0,3/2)$, $0<\alpha_1< \min (\beta/2, 1/8, (3-2\beta)/16)$ and let $w$ satisfy Assumption \ref{ass:w}. Let $\varphi(t)$ be the solution of  \eqref{NLS} on the maximal time interval $[0,\Tmax)$ as in  Lemma \ref{lem:NLS} with initial datum $\varphi_0\in H^4(\R^2)$, $\norm{\varphi_0}_{L^2}=1$. Let $\Psi_{N}(t)$ be the solution of \eqref{eq:SE} with a normalized initial state $\Psi_{N,0} \in L^2_s(\R^{2N})$ satisfying 
\begin{align} \label{eq:initial-Phi0-HN}
N\Tr \Big((1-\Delta) \,q\, \gNoz\,q\, \Big) \leq C, \quad q= 1 - |\varphi_0\rangle \langle \varphi_0|
\end{align}
for some constant $C>0$.
Then  for every $t\in[0,\Tmax)$, we have Bose--Einstein condensation in the state $\varphi(t)$, i.e.,
\begin{equation} \label{eq:1pdm-cv-phi}
\Tr\left| \gNot-|\varphi(t)\rangle\langle\varphi(t)|\right|\leq C_t N^{-\alpha_1}
\end{equation}
for sufficiently large $N$, where $C_t$ is independent of $N$ and continuous on $[0,\Tmax)$. 
\end{theorem}

The initial condition \eqref{eq:initial-Phi0-HN} means that at the time $t=0$, the total kinetic energy of all excited particles outside the condensate $\varphi_0$ is bounded.  Thus,  there are only few excitations, which is a key assumption allowing us to control the fluctuations around the condensate $\varphi(t)$ for all $t\in[0,\Tmax)$ by using an energy method. The kinetic bound \eqref{eq:initial-Phi0-HN} has been proven for the ground state or low-lying excited states of trapped systems with suitable repulsive interactions, see e.g. \cite{Seiringer-11,LewNamSerSol-15}.

The following statement is a direct consequence of Theorem \ref{thm:1} and the definition of $\Tmax$ in Lemma  \ref{lem:NLS}.

\begin{cor}[Many-body blow up]\label{cor:RDM} We keep the same assumptions as in Theorem~\ref{thm:1}, and assume additionally that $\Tmax<\infty$. Then there exists a sequence $N(t)\in \mathbb{N}$ such that $N(t)\to\infty$ as $t\nearrow\Tmax$ and such that
\begin{equation}\label{eq:many-body-blow-up}
\lim\limits_{t\nearrow\Tmax} \frac{1}{N(t)}\lr{\Psi_{N(t)}(t), \sum_{j=1}^{N(t)}( -\Delta_j) \Psi_{N(t)}(t)}= \lim\limits_{t\nearrow\Tmax} \Tr (-\Delta \gamma^{(1)}_{\Psi_N(t)}) =\infty. 
\end{equation}
\end{cor}

The implication of Corollary \ref{cor:RDM} follows from a  well-known argument (see \cite[Remark 2]{MicSch-10}): for every $t\in[0,\Tmax)$, the trace convergence in Theorem \ref{thm:1} and Fatou's lemma imply that
\begin{equation}\label{eq:Fatou}
\liminf\limits_{N\to\infty}\Tr\left((1-\Delta)\gNo(t)\right)\geq \Tr\left((1-\Delta)|\varphi(t)\rangle\langle\varphi(t)|\right)=\norm{\varphi(t)}^2_{H^1(\R^2)}.
\end{equation}
Therefore, if $\Tmax<\infty$, then the one-body blow-up condition \eqref{eq:blow-up} implies the many-body blow-up result \eqref{eq:many-body-blow-up}. Note that \eqref{eq:Fatou} is only an inequality,  hence the reverse direction, which would imply that the many-body blow-up phenomenon does not occur at any fixed time $t\in[0,\Tmax)$, cannot be deduced from Theorem \ref{thm:1}. We expect that this  holds true, but a proof would require some additional analysis, which we will not pursue in the present work.

Note that in Theorem \ref{thm:1} we do not make any assumption on the sign of the potential $w$, hence our result also covers the repulsive case $w\ge 0$. In this case, Theorem \ref{thm:1} holds true for all $\beta>0$; more precisely, it was proved in \cite{JebLeoPic-16} for the scaling regime $e^{2N} w(e^N x)$, which leads to a subtle correction producing the scattering length of the potential in the NLS \eqref{NLS} instead of $b=\int w$.

Our result is mostly interesting in the focusing case $w\le 0$. In this case, the NLS \eqref{NLS} was derived in \cite{NamNap-19} under the stability condition  $\int |w_-| < a^*$ for $\beta\in (0,1)$ (see also \cite{CheHol-15,JebPic-18} for earlier related results). 
Without the stability condition, the derivation of the NLS \eqref{NLS} prior to the blow-up time can be shown for $\beta<1/2$, following the methods in \cite{Pickl-10,NamNap-17a}, given the uniform-in-$N$ bounds on the Hartree equation which we prove in Lemma \ref{lem:hartree}. 
Thus, our result, which applies to all $\beta\in (0,3/2)$ and to a general $w$, substantially extends these existing results.

\subsection{Norm approximation}

Let us now discuss the fluctuations around the condensate. For this purpose, we first introduce
the following Hartree-type equation
\begin{equation}\label{hartree}
\begin{cases}
\i\partial_t u_N(t,x)= \left(-\Delta_x+(\wN*|\uN(t,\,\cdot\,)|^2)(x)-\mN(t)\right)u_N(t,x) =:\hN(t) \uN(t,x),\\
u_N(0,x)= \varphi_0(x),
\end{cases}
\end{equation}
with
\begin{equation}\label{mN}
\mN(t)=\frac12\iint_{\R^2\times\R^2} |\uN(t,x)|^2\wN(x-y)|\uN(t,y)|^2\dx\dy.
\end{equation}
The Hartree dynamics \eqref{hartree} essentially play the same role as the NLS dynamics \eqref{NLS} in the leading order  description, but using the former is slightly more natural for the second order approximation (see \cite{LewNamSch-15,NamNap-17,NamNap-17a,BreNamNapSch-19,NamNap-19} for a similar choice). In particular, \eqref{hartree} has a unique global solution, and $\|u_N(t)\|_{H^1}$ is bounded uniformly in $N$ and locally in time when $t\in [0,\Tmax)$ with $\Tmax$ given in Lemma \ref{lem:NLS}. Moreover, since $u_N(t)\to \varphi(t)$ in $L^2(\R^2)$ as $N\to \infty$, the convergence \eqref{eq:1pdm-cv-phi} remains true if $\varphi(t)$ is replaced by $u_N(t)$   (see Lemma \ref{lem:hartree} for the details). 

To describe the excitations around the condensate, it is convenient to switch to a Fock space setting where the number of particles is not fixed. Let us introduce the one-body excited space
\begin{equation}
\fHp(t)=\left\{\uN(t)\right\}^\perp\subset \fH= L^2(\R^2)
\end{equation}
and the (bosonic) Fock spaces over $\fHp$,
\begin{equation}\label{Fock:space}
\FNp(t)=\bigoplus_{k=0}^N\bigotimes_\sym^k \fHp
\;\subset \Fp(t)=\bigoplus_{k\geq 0}\bigotimes_\sym^k \fHp
\;\subset \;\Fock=\bigoplus_{k\geq 0}\bigotimes_\sym^k \fH\,.
\end{equation}
Note that $\Fp(t)$ and its subspace $\Fp^{\le N}(t)$ are time-dependent via $\uN(t)$, and they are naturally embedded in the full Fock space $\Fock$ over $\fH$. 

Let us recall the standard second quantization formalism, where the creation and annihilation operators on $\Fock$, $\ad(f)$ and $a(f)$, are defined by 
\begin{subequations}
\begin{align}
(\ad(f)\Chi)^{(k)}(x_1\mydots x_k) &= \frac{1}{\sqrt{k}}\sum\limits_{j=1}^kf(x_j)\chi^{(k-1)}(x_1\mydots x_{j-1},x_{j+1}\mydots x_k),\quad \forall k\ge 1, \\
(a(f)\Chi)^{(k)}(x_1\mydots x_k) &= \sqrt{k+1}\int_{\R^2}\d x\overline{f(x)}\chi^{(k+1)}(x_1\mydots x_k,x), \quad \forall k\ge 0
\end{align} 
\end{subequations}
for all $f\in L^2(\R^2)$ and $\Chi=(\chi^{(k)})_{k=0}^\infty \in \Fock$. It is also convenient to introduce the operator-valued distributions  $\ad_x$, $a_x$ by 
\begin{equation}
\ad(f)=\int\d x f(x)\,\ad_x\,,\qquad a(f)=\int\d x\overline{f(x)}\,a_x\,,
\end{equation}
which satisfy the canonical commutation relations
\begin{equation}
[a_x,\ad_y]=\delta(x-y)\,,\qquad [a_x,a_y]=[\ad_x,\ad_y]=0\,.
\end{equation}
Using this language, we define the second quantization of one- and two-body operators as 
\begin{align}\label{def:dGot}
\dGo(T)&=0\oplus\bigoplus_{k\geq 1}\sum_{j=1}^k T_j=\iint T(x;x')\ad_xa_{x'}\dx\dx'\,,\\
\dGt(S)&=0\oplus0\oplus\bigoplus_{k\geq 2}\sum_{1\leq i<j\leq k}S_{ij}
=\frac12\iiiint S(x,y;x',y')\ad_x\ad_ya_{x'}a_{y'}\dx\dy\dx'\dy'\,, \nonumber
\end{align}
for $T(x,x')$ and $S(x,y;x,y')$ the kernels of the operators $T$ on $\fH$ and $S$ on $\fH^2$ (see, e.g., \cite[Section 7]{Solovej-ESI-2014}).
In this language, the Hamiltonian~\eqref{HN} can be expressed equivalently as 
\begin{equation}\label{eq:HN}
H_N = \dG_1(-\Delta) + \frac{1}{N-1} \dG_2 (w_N)
\end{equation}
on $\fH^N$. We also introduce the number operator on Fock space  $\Fock$,  
\begin{equation}
\Number=\dG_1(\id), 
\end{equation}
where $\id$ this is the identity operator on $\fH$, and denote the cut-off functions
\begin{equation}
\idm=\id(\Number \leq m)\,,\qquad \idgm=\id(\Number > m), \quad \forall m\in (0,\infty). 
\end{equation}

Following the approach in \cite{LewNamSerSol-15,LewNamSch-15}, the $N$-body dynamics $\Psi_N(t)\in L^2_{s}(\R^{2N})$  can be decomposed as 
\begin{equation}\label{eqn:decomposition:PsiN}
\Psi_N(t)=\sum\limits_{k=0}^N{\uN(t)}^{\otimes (N-k)}\otimes_s\phi_N^{(k)}(t) =  \sum\limits_{k=0}^N \frac{ \ad(\uN(t))^{\otimes (N-k)}}{\sqrt{(N-k)!}}\phi_N^{(k)}(t)
\end{equation}
for $\otimes_s$ the symmetric tensor product and
where the vector 
\begin{equation}
\PhiN(t)=\big(\phi_N^{(k)}(t)\big)_{k=0}^N \in \FNp(t) \subset \Fp(t)
\end{equation}
describes the excitations around the condensate $u_N(t)$ (see Section \ref{sec:strategy} for details). 

Our goal is to approximate the $N$-body dynamics $\PhiN(t)$  by the solution  $\Phi(t)$ of the simpler evolution equation
\begin{equation}\label{eq:Bog}
\begin{cases}
&\i\partial_t\Phi(t)=\bH(t)\Phi(t)\\
&\Phi(0)=\Phi_0\,,
\end{cases}
\end{equation}
where $\bH(t)$ denotes the Bogoliubov-Hamiltonian
\begin{equation}\label{H:Bog}
\bH(t)= \dGo( \hN(t)+ K_1(t)) +\frac12\left(\iint K_2(t,x,y)\ad_x\ad_y\dx\dy+\hc\right). 
\end{equation}
 In \eqref{H:Bog}, $\hN(t)$ is defined in the Hartree equation \eqref{hartree}, and 
\begin{equation}\label{K_1:K_2}
K_1(t)=q(t)\tilde{K}_1(t)q(t)\,,\qquad K_2(t)=q(t)\otimes q(t) \tilde{K}_2(t)
\end{equation}
where 
\begin{equation}\label{def:p:q}
q(t)=1-p(t)=1- |\uN(t)\rangle\langle\uN(t)|
\end{equation} 
and the kernel of the operator $\tilde{K}_1(t)$ and the function $\tilde{K}_2(t)\in\fH^2$ are given by
\begin{equation}\begin{split}
\tilde{K}_1(t,x,y)&=\uN(t,x)\wN(x-y)\overline{\uN(y)}\,,\\
\tilde{K}_2(t,x,y)&=\uN(t,x)\wN(x-y)\uN(y)\,.
\end{split}\end{equation}

The effective generator $\bH(t)$ emerges from the Bogoliubov approximation when we write the Hamiltonian $H_N$ in the second quantization formalism, then implement the c-number substitution $a(u_N),\ad(u_N)\mapsto \sqrt{N}$, and finally keep only the terms that are quadratic in creation and annihilation operators. Note that $\bH(t)$ is an operator on the full Fock space $\Fock$ since $\hN(t)$ does not leave $\fHp(t)$ invariant, but it does not contradict the fact that  $\Phi(t)\in\Fp(t)$ (see e.g. \cite{LewNamSch-15} for a detailed explanation).
Moreover, $\bH(t)$  is $N$-dependent, although we do not make this explicit in the notation. The Bogoliubov equation \eqref{eq:Bog} is globally well-posed (see Lemma~\ref{lem:Bog}). 

Now we are ready to state our second main result.

\begin{theorem}[Bogoliubov excitations from the condensate]\label{thm:2} Let $\beta\in (0,3/2)$, $0<\alpha_2< \min (1/8, (3-2\beta)/16)$ and let  $w$ satisfy Assumption \ref{ass:w}. Let $u_N(t)$ be the solution of the Hartree equation \eqref{hartree} with initial darum $\varphi_0\in H^4(\R^2)$, $\|\varphi_0\|=1$.  Let $\Phi(t)=(\phi^{(k)}(t))_{k=0}^\infty \in\Fp(t)$ be the solution of the Bogoliubov equation \eqref{eq:Bog} with initial datum $\Phi_0=(\phi^{(k)}_0)_{k=0}^\infty\in \Fp(0)$ satisfying $\|\Phi_0\|=1$ and
\begin{align}  \label{eq:initial-Phi0}
\langle \Phi_0, \dGo (1-\Delta) \Phi_0\rangle \leq C
\end{align}
for some constant $C\geq0$. 
Let $\Psi_N(t)$ the the solution of the Schr\"odinger equation \eqref{eq:SE} with initial datum 
\begin{equation} \label{eq:initial-PhiN0}
\Psi_{N,0}=\sum\limits_{k=0}^N{\varphi_0}^{\otimes (N-k)}\otimes_s \phi_0^{(k)}=  \sum\limits_{k=0}^N \frac{ \ad(\varphi_0)^{\otimes (N-k)}}{\sqrt{(N-k)!}} \phi_0^{(k)}.
\end{equation}
Then, for all $t\in [0,\Tmax)$, we have the norm approximation
\begin{equation} \label{eq:norm-PhiNt}
\Big\|\Psi_N(t)-\sum_{k=0}^N \uN(t)^{\otimes (N-k)}\otimes_s\phi^{(k)}(t)\Big\| \leq C_{t}  N^{-\alpha_2}\,,
\end{equation}
where the constant $C_t$ is independent of $N$ and continuous in $t\in [0,\Tmax)$. 
\end{theorem}

Note that under the decomposition \eqref{eq:initial-PhiN0}, the kinetic condition \eqref{eq:initial-Phi0-HN} is equivalent to condition \eqref{eq:initial-Phi0} in Theorem \ref{thm:2} (see Remark \ref{rmk:UN-PhiN}). 
Strictly speaking, the state $\Phi_{N,0}$ in \eqref{eq:initial-PhiN0} is not normalized in $L^2_s(\R^{2N})$, but 
the condition \eqref{eq:initial-Phi0} ensures that
\begin{align}
1\ge \|\Psi_{N,0}\|^2 = 1-  \| \id_{\{\cN> N\}}  \Phi_0\|^2 \ge 1-\langle \Phi_0, (\cN/N) \Phi_0\rangle \ge 1- CN^{-1}.
\end{align}

The norm approximation in $N$-body space  was given in \cite{LewNamSch-15} in the mean-field regime ($\beta=0$). Recently, also higher order corrections to Bogoliubov's theory in the mean-field regime were derived in \cite{BosPetPicSof-21}.  
For repulsive interaction $w\ge 0$, the validity of Bogoliubov's theory was extended to $0<\beta<1$ in 3D \cite{BreNamNapSch-19} (see also  \cite{NamNap-17a,NamNap-17} for earlier results). In 3D  it is natural to restrict $0<\beta<1$ as Bogoliubov's theory is no longer correct in the Gross--Pitaevskii regime, but the method in  \cite{BreNamNapSch-19}  seems applicable to a less restricted range of $\beta$ in 2D. Hence, our result is mainly interesting in the attractive case $w\le 0$, where the validity of Bogoliubov's theory was known only for $0<\beta<1$ in the stability regime $\int w_- > -a^*$ \cite{NamNap-19}. 

Note that in the literature, there are also many works devoted to the dynamics around the coherent states in Fock space, initiated in \cite{Hepp-74,GinVel-79b, GriMacMar-10,GriMacMar-11,BocCenSch-17}. We refer to \cite{LewNamSch-15} for a detailed comparison between the $N$-body setting and the Fock space situation, and the book \cite{BenPorSch-15} for further results. Our method is also applicable to this setting, but we skip the details.  

The ideas of our proof are explained in the next section, and the full technical details are provided afterwards. 

\subsection*{Notation}
\begin{itemize}
\item We will use $C>0$ for a general constant which  may depend on $w$ and $\varphi_0$ and which may vary from line to line. We also use the notation $C_{t}$ to highlight the time dependence. 
\item When it is unambiguous, we abbreviate the  $L^2$-norm and the corresponding inner product by $\norm{\,\cdot\,}$ and $\lr{\,\cdot\,,\,\cdot\,}$, respectively.
\end{itemize}

\section{Proof strategy}\label{sec:strategy}

In this section we explain the main ingredients of the proof. We will focus on Theorem \ref{thm:2}, which implies Theorem \ref{thm:1}. Our approach is based on Bogoliubov's approximation where the fluctuations around the condensate are effectively described by an evolution equation with a quadratic generator in Fock space. The main mathematical challenge is to justify this approximation by rigorous estimates. Let us first give an overview of the proof strategy, and then we come to the detailed setting.

As an important input of Bogoliubov's theory \cite{Bogoliubov-47}, we expect that most particles are in the condensate $\uN(t)$, which is governed by the Hartree equation \eqref{hartree}. The first step in our analysis is to establish several uniform-in-$N$ bounds for the Hartree dynamics, which is nontrivial due to the instability issue. These one-body estimates require a careful adaptation of the analysis of the NLS \eqref{NLS} in \cite{Cazenave}, which will be discussed in Section \ref{sec:Hartree}. In the following, we will focus on the many-body aspects of the proof. 

In order to extract the excitations, namely  the  particles outside the condensate, from the $N$-body wave function $\Psi_N(t)$, we use the unitary transformation $U_N(t)$ introduced in  \cite{LewNamSerSol-15}. This is a mathematical tool to implement Bogoliubov's c-number substitution \cite{Bogoliubov-47}, resulting in the evolution $\Phi_N(t)=U_N(t) \Psi_N(t)$ on the excited Fock space $\FNp(t)$ where the generator $\GN(t)$ was computed explicitly in \cite{LewNamSch-15}. Thus,  we can rewrite \eqref{eq:norm-PhiNt} in terms of excitations as 
\begin{align} \label{eq:PhiN-Phi-intro}
\| \Phi_N (t) - \Phi(t)\|^2 \le C_t N^{-\alpha_2}
\end{align}
for all $t\in[0,\Tmax)$, where $\Phi(t)$ is the solution to the Bogoliubov equation \eqref{eq:Bog}. 

The main difficulty in proving \eqref{eq:PhiN-Phi-intro} is the lack of the stability of the second kind \eqref{eq:HN-CN}. More precisely, with an arbitrarily negative potential $w$, we do not expect to have a good lower bound for the generator $\GN(t)$ of $\Phi_N(t)$, which in turn prevents us from obtaining a good kinetic bound for $\Phi_N(t)$. A key observation in \cite{NamNap-19} is that a weaker version of the stability \eqref{eq:HN-CN} holds if we restrict to a space of few excitations. Rigorously, for the truncated dynamics $\Phi_{N,M}(t) \in \cF_{\bot}^{\le M}(t)$ which is associated to the generator  $\idM \GN(t) \idM$ with a parameter $M=N^{1-\delta}$, $\delta\in (0,1)$, it was proved in  \cite{NamNap-19} that $\Phi_{N,M}$ satisfies an essentially-uniform kinetic bound, and hence $\|\Phi_{N,M}(t)-\Phi(t)\|$ can be controlled efficiently (see Lemma \ref{lem:Bog} below). 

Thus,  by the triangle inequality, the main missing ingredient for \eqref{eq:PhiN-Phi-intro} is a good estimate for the norm $\|\Phi_N(t)-\Phi_{N,M}(t)\|$. For this term, we cannot use the analysis in \cite{NamNap-19},  which crucially relies on the stability condition $\int |w_-| < a^*$. The main novelty of the present paper is the introduction of  a new method which does not require any information about the full dynamics $\Phi_{N}$. This kind of ideas was previously used in \cite{NamNap-17a}, where various propagation bounds were established by Cauchy--Schwarz inequalities of the form 
\begin{align} \label{eq:PhiN-Phi-intro-1}
|\langle \Phi_N, A \Phi_{N,M}\rangle | \le \| \Phi_N \| \| A \Phi_{N,M}\|.
\end{align}
However, this approach is insufficient to handle the dilute regime where $\beta>1/2$. To improve the Cauchy--Schwarz argument, we decompose $1=\cW^{-1}\cW$ with a suitable weight $\cW>0$  and split  \eqref{eq:PhiN-Phi-intro-1} into  
\begin{align} \label{eq:PhiN-Phi-intro-2}
|\langle \Phi_N, A \Phi_{N,M}\rangle | &\le  |\langle \Phi_N, \cW^{-1} A  \cW \Phi_{N,M}\rangle |  + |\langle \Phi_N,  \cW^{-1} [\cW,  A] \Phi_{N,M}\rangle |.
\end{align}
The first term on the right-hand side of \eqref{eq:PhiN-Phi-intro-2} looks similar to $|\langle \Phi_N, A \Phi_{N,M}\rangle |$ but it is easier to bound by the Cauchy--Schwarz inequality provided that we can bound $\|A^* \cW^{-1}\Phi_N \|$ in terms of  $\|\Phi_N\|$ in an average sense. For the second term on the right-hand side, we gain some cancelation due to the commutator  $[\cW,  A]$, which eventually ensures that $\|\cW^{-1} [\cW,  A] \Phi_{N,M}\|$ is much smaller than $\|A \Phi_{N,M}\|$. 

In the remainder of this section, we will provide some further details of the above ingredients.

\subsection{Reformulation of the Schr\"odinger equation}

Our starting point is a reformulation of the Schr\"odinger equation \eqref{eq:SE},  following the method  proposed in \cite{LewNamSerSol-15,LewNamSch-15}. 

Let $u_N(t)$ be the Hartree evolution in \eqref{hartree}. To factor out the contribution of the condensate, we use the excitation map $U_N(t):\fH^N(t)\to\FNp(t)$ defined by 
\begin{equation}
U_N(t)=\bigoplus_{k=0}^N q(t)^{\otimes k} \left( \frac{a(u_N(t))^{N-k}}{\sqrt{(N-k)!}} \right)
\end{equation}
where
$
q(t)   = 1- |\uN(t)\rangle\langle\uN(t)|
$
as in \eqref{def:p:q}.
It was proven in  \cite{LewNamSerSol-15} that $U_N(t)$ is a unitary transformation and its inverse is given by \eqref{eqn:decomposition:PsiN}, namely
$$
U_N(t)^* \Phi  = \sum\limits_{k=0}^N{\uN(t)}^{\otimes (N-k)}\otimes_s\phi^{(k)} =  \sum\limits_{k=0}^N \frac{ \ad(\uN(t))^{\otimes (N-k)}}{\sqrt{(N-k)!}}\phi^{(k)}
$$
for all $\Phi=(\phi_k)_{k=0}^N \in \FNp(t)$. Heuristically, the mapping $U_N$ provides an efficient way of focusing on the fluctuations around the Hartree state $u_N(t)^{\otimes N}$; in particular,  $U_N(t) u_N(t)^{\otimes N}=\Omega$ is the vacuum of $\Fp(t)$. 

It was also proven in \cite{LewNamSerSol-15}  that for $f,g\in\fHp(t)$, we have the following identities on $\FNp(t)$: 
\begin{align}\label{eqn:substitution:rules}
U_N\, \ad({\uN})a({\uN})U_N^*&=N-\Number\,, \nn\\
U_N\, \ad(f)a({\uN}) U_N^*&=\ad(f)\sqrt{N-\Number}\,, \nn\\
U_N\, \ad({\uN})a(g)U_N^*&=\sqrt{N-\Number}a(g)\,,\\
U_N\, \ad(f)a(g)U_N^*&=\ad(f)a(g)\,.\nn
\end{align}
If Bose--Einstein condensation holds, then, in an average sense, $\cN \ll N$ in $\FNp(t)$. Therefore, \eqref{eqn:substitution:rules} can be interpreted as a rigorous implementation of Bogoliubov's c-number substitution \cite{Bogoliubov-47}, where $a(\uN)$ and $\ad(\uN)$ are formally replaced by the scalar number $\sqrt{N}$.  

\begin{remark} \label{rmk:UN-PhiN} Note that from \eqref{eqn:substitution:rules} we have $U_N^*\dG_1(qAq)U_N=\dG_1(qAq)$ for any operator $A$ on $\fH$.  Consequently, \eqref{eq:initial-Phi0-HN} is  equivalent to  \eqref{eq:initial-Phi0}.
\end{remark}

Now we consider the transformed dynamics 
\begin{align} \label{eq:PhiN}
\Phi_N(t)=U_N(t)\Psi_N(t).
\end{align}
The Schr\"odinger equation \eqref{eq:SE}  can be written in the equivalent form
\begin{equation}\label{SE:exc}
\begin{cases}
&\i\partial_t\PhiN(t)=\GN(t)\PhiN(t)\\
&\PhiN(0)= U_N(0)^* \Psi_{N,0}. 
\end{cases}
\end{equation}
Here, the generator $\GN(t)$ can be computed explicitly, using the second-quantized form \eqref{eq:HN} and the rules \eqref{eqn:substitution:rules} (see 
\cite[Appendix B]{LewNamSch-15}), as
\begin{equation}
\GN(t)=\left(\i\partial_tU_N(t)\right)U_N^*(t)+U_N(t)H_NU_N^*(t)=\frac12\sum_{j=0}^4\idN\left(\bG_j+\bG_j^*\right)\idN \label{eqn:GN}
\end{equation}
with 
\begin{subequations}
\begin{align} 
\bG_0&= \dGo(\hN)+\dGo(K_1)\frac{N-\Number}{N-1}+\dGo\left(q(t)(\hN+\Delta)q(t)\right)\frac{1-\Number}{N-1}\,,\\
\bG_1&=-2\ad \left(q(t)(\wN*|\uN(t)|^2)\uN(t)\right) \frac{\Number\sqrt{N-\Number}}{N-1} \,,\\
\bG_2&=\iint K_2(t,x,y)\ad_x\ad_y\dx\dy\frac{\sqrt{(N-\Number)(N-\Number-1)}}{N-1}\,,\\
\bG_3&=\iiiint\left(q(t)\otimes q(t) \wN \id\otimes q(t)\right)(x,y;x',y'){\uN(t,x)}\nonumber\\
&\qquad\qquad\times \ad_x \ad_y a_{y'} \dx\dy\dx'\dy'\ \frac{\sqrt{N-\Number}}{N-1}, \\
\bG_4&=\frac{1}{N-1}\dG_2 \left(q(t)\otimes q(t)\wN q(t)\otimes q(t)\right).
\end{align}
\end{subequations}
Recall that $\hN(t)$ is given in \eqref{hartree}, and $K_1(t)$ and $K_2(t)$ are given in \eqref{K_1:K_2}. In the above notation, $\wN$ denotes the function $\wN:\R^2\to\R$ in $\bG_1(t)$, and the two-body multiplication operator $\wN(x-y)$ in $\bG_3(t)$ and $\bG_4(t)$.

\subsection{Simplified dynamics} 

Following Bogoliubov's heuristic ideas \cite{Bogoliubov-47}, we consider a simplification of \eqref{SE:exc}, where only the quadratic terms $\bG_0$ and $\bG_2$ in the generator are kept. This leads to the Bogoliubov equation  \eqref{eq:Bog}, whose well-posedness is well-known, see for example \cite[Lemma 5]{NamNap-19}. 

\begin{lem}[Bogoliubov dynamics]\label{lem:Bog} 
Let  $w$ satisfy Assumption \ref{ass:w}, let $u_N(t)$ be the Hartree evolution in \eqref{hartree} with initial state $\varphi_0\in H^4(\R^2)$, and let $\Phi_0\in \Fp(0)$ be a normalized vector satisfying \eqref{eq:initial-Phi0}. Then the Bogoliubov equation  \eqref{eq:Bog} with initial condition $\Phi_0$ has a unique global solution
$$\Phi\in C\big( [0,\infty),\Fock\big) \cap L^\infty_\mathrm{loc}\big( (0,\infty),\mathcal{Q}(\dGo(1-\Delta))\big)$$
and
 $\Phi(t)\in\Fp(t)$ for all $t>0$. Moreover, for all $t\in[0,\Tmax)$ and any $\eps>0$, we have 
\begin{equation} \label{eq:Bog-kinetic-bound}
\lr{\Phi(t),\dGo(1-\Delta)\Phi(t)}\leq C_{t,\varepsilon} N^\varepsilon\,,
\end{equation}
where $\Tmax$ is given in Lemma \ref{lem:NLS}  and the constant $C_{t,\eps}$ is independent of $N$. 
\end{lem}

\begin{proof}
The global well-posedness of $\Phi(t)$ is shown in \cite[Theorem 7]{LewNamSch-15}. The kinetic bound \eqref{eq:Bog-kinetic-bound} follows from the analysis in  \cite[Lemma 5]{NamNap-19} and the uniform bounds of $u_N(t)$, which will be given later in Lemma \ref{lem:hartree}. \end{proof}

In order to estimate the difference $\norm{\PhiN(t)-\Phi(t)}$, we follow \cite{NamNap-19} and introduce the truncated dynamics $\PhiNM(t) \in \cF_{\bot}^{\le M}(t)$, which solve the equation 
\begin{equation}\label{eq:SE:truncated}
\begin{cases}
&\i\partial_t\PhiNM(t)=\idM\GN(t)\idM\PhiNM(t)\\
&\PhiNM(0)=\idM\Phi_0\,.
\end{cases}
\end{equation}
As explained in \cite{NamNap-19}, the main advantage of \eqref{eq:SE:truncated} is that the truncated generator is stable, namely 
\begin{equation}
\idM\GN(t)\idM \ge \frac{1}{2} \dG_1(-\Delta) - C_{t,\eps} N^\eps 
\end{equation}
for all $t\in[0,\Tmax)$. This allows to establish an efficient kinetic bound for $\PhiNM(t)$, which is not available for $\Phi_N$. Consequently, it is much easier to compare $\PhiNM(t)$ with the Bogoliubov dynamics. We collect some known properties of $\PhiNM(t)$ in the following lemma.

\begin{lem}[Truncated dynamics]\label{lem:nam} We keep the assumptions of Lemma \ref{lem:Bog}. Let $M=N^{1-\delta}$ for some constant $\delta\in (0,1)$. Then  the equation \eqref{eq:SE:truncated} has a unique global solution $\PhiNM(t) \in \cF_{\bot}^{\le M}(t)$ with $t\in [0,\infty)$. Moreover, for every $t\in[0,\Tmax)$, we have 
\begin{equation} \label{eq:nam:11}
\lr{\PhiNM(t),\dGo(1-\Delta)\PhiNM(t)}\leq C_{t,\varepsilon}N^\varepsilon\,
\end{equation}
and 
\begin{equation}\label{nam:15}
\norm{\PhiNM(t)-\Phi(t)}^2\leq C_{t,\varepsilon}N^\varepsilon\left(\sqrt{\frac{M}{N}}+\frac{1}{M}\right). 
\end{equation}
\end{lem}

\begin{proof} The global well-posedness of $\PhiNM(t)$ follows from the general method in \cite[Theorem 7]{LewNamSch-15} (see also \cite[Section 6]{NamNap-19}).  Given the uniform bounds of $u_N(t)$ in Lemma \ref{lem:hartree}, the bounds \eqref{eq:nam:11} and \eqref{nam:15} follow from the arguments in \cite[Lemma 11]{NamNap-19} and \cite[Lemma 15]{NamNap-19}, respectively. 
\end{proof}

\subsection{From the truncated to the full dynamics} \label{sec:stratefy-3} Given Lemma \ref{lem:nam}, the missing piece for the proof of Theorem \ref{thm:2} is an estimate for $\norm{\PhiN(t)-\PhiNM(t)}$. The main new ingredient of the present  paper is the following bound:

\begin{proposition}\label{prop:B} We keep the assumptions of Lemma \ref{lem:Bog}. Let $M=N^{1-\delta}$ for some constant $\delta\in (0,1)$. Let $\Phi_N$ and $\Phi_{N,M}$ be solutions of \eqref{SE:exc} and \eqref{eq:SE:truncated}, with initial data $\Phi_N(0)=\id^{\le N} \Phi_0$, $\Phi_{N,M}(0)=\id^{\le M} \Phi_0$, respectively. 
Then for every $t\in[0,\Tmax)$ and every $\eps>0$, we have 
\begin{align} \label{eq:cB-final}
\norm{\PhiN(t)-\PhiNM(t)}^2 \le C_{t,\eps} N^\eps \left(\frac{1}{\sqrt{M}}+\frac{N^\beta}{M^{3/2}}\right).
\end{align}
\end{proposition}

Eventually, we will take $\delta>0$ small, hence the condition $\beta<3/2$ is needed to ensure that the error term $N^{\beta}/M^{3/2}$ on the right hand side of \eqref{eq:cB-final} is negligible.

In order to prove Proposition \ref{prop:B}, by norm conservation of $\|\Phi_N(t)\|$ and $\|\Phi_{N,M}(t)\|$, it suffices to show that $\lr{\PhiN(t),\PhiNM(t)}$ is close to $1$. For technical reasons, it is more convenient to consider $\lr{\PhiN(t),f_M^2\PhiNM(t)}$ with $f_M$ a smoothened version $\id^{\le M}$. To be precise, we fix a smooth function $f:\R\to[0,1]$  such that $f(s)=1$ for $s\le 1/2$ and $f(s)=0$ for $s\ge 1$, and define the operator $\fM$ on $\Fock$ by
\begin{equation} \label{eq:fM}
\fM=f\left(\frac{\Number}{M}\right).
\end{equation}

We will deduce Proposition \ref{prop:B} from a Gr\"onwall argument and the estimate
\begin{equation} \label{eq:probB-simplified}
\left|\frac{\d}{\d t} \lr{\PhiN(t),f_M^2\PhiNM(t)}  \right|\le C_{t,\varepsilon}N^\varepsilon\left(\frac{1}{\sqrt{M}}+\frac{N^{\beta}}{M^{3/2}}\right).  
\end{equation}

It remains to explain the proof of \eqref{eq:probB-simplified}. Let us drop the time dependence from the notation where it is unambiguous. From the equations  \eqref{SE:exc} and \eqref{eq:SE:truncated},  we have 
\begin{align}\label{eq:quadratic-generator}
\left|\frac{\d}{\d t} \lr{\PhiN(t),f_M^2\PhiNM(t)}  \right|= \left|\Im\lr{\PhiN,[\GN,\fM^2]\PhiNM}\right| 
\end{align}
since $\PhiNM\in\FMp$ and $\fM^2\idM=\fM^2$. Then it is straightforward to decompose $\GN$ into the sum of $\bG_j$ as in \eqref{eqn:GN}. Since $f_M$ is a function of $\cN$, only the particle number non-preserving terms $\bGo$, $\bGt$ and $\bGth$ contribute to the commutator. 

One of the most difficult terms is the quadratic one $\left|\lr{\Phi_N,[\bG_2,\fM^2]\PhiNM}\right|$, where two annihilation operators hit $\Phi_N$. Since $\bG_2$ only changes the number of particles by at most 2, the commutator with $\fM^2$ allows us to gain a factor $M^{-1}$. Therefore,  
estimating \eqref{eq:quadratic-generator} essentially boils down to proving a bound for
\begin{equation}\label{eqn:quadratic:intro}
\frac{1}{M}\left|\iint\dx\dy\, \uN(x)\uN(y)\wN(x-y)\lr{\Phi_N,\ad_x\ad_y\PhiNM}\right|\,.
\end{equation}

In \cite{NamNap-19}, a variant of this term was estimated using a kinetic bound for $\Phi_N$  based on the method in \cite{Lewin-15} and the stability condition $\int |w_-| < a^*$.  In the present paper, since we are considering a general potential $w$ including the instability regime $\int w_- < - a^*$, we only have 
\begin{equation}
\langle \Phi_N, \dG_1 (1-\Delta)\Phi_N \rangle \le C N^{1+2\beta}, 
\end{equation}
which can be deduced from a variant of the energy lower bound \eqref{eq:HN-CN-beta}. However, the latter bound is too weak, and inserting it in the analysis in \cite{NamNap-19} produces a solution only for $\beta<1/2$. 

Another idea, which can be extracted from the approach in \cite{NamNap-17a}, is to handle \eqref{eqn:quadratic:intro} by the Cauchy--Schwarz inequality 
\begin{equation} \label{eq:simple-CS}
|\langle \Phi_N, \ad_x \ad_y \Phi_{N,M}\rangle| \le \|\Phi_N\| \| \ad_x \ad_y \Phi_{N,M}\|.
\end{equation}
(To be precise, a variant of this argument was used in \cite{NamNap-17a} to compare $\Phi_N$ directly with the Bogoliubov dynamics $\Phi$.) The advantage of \eqref{eq:simple-CS} is that no information about $\Phi_N$ is needed. However, since we have to couple \eqref{eq:simple-CS} with the singular potential $\wN(x-y)$ in \eqref{eqn:quadratic:intro}, we eventually obtain a large factor $\|w_N\|_{L^\infty} \sim N^{2\beta}$, and the final bound is only good for $\beta<1/2$. 

Thus,  to cover the extended range $\beta\in (0,3/2)$, new ideas are needed to handle \eqref{eqn:quadratic:intro}. In the present paper, on the one hand, we will not  rely on  any information of $\Phi_N$; moreover, instead of using directly \eqref{eq:simple-CS} we will further decompose \eqref{eqn:quadratic:intro} by introducing a weight given by
\begin{equation}\label{def:R}
\cR:=\dGt(|\wN|)+1 = \frac12\int\dx\dy|\wN(x-y)|\ad_x\ad_ya_xa_y+1.
\end{equation}
By inserting $1=\cR^{-\frac12}\cR^\frac12$, we can split 
\begin{align}
\ad_x \ad_y = \cR^{-1/2}  \ad_x \ad_y \cR^{1/2} +  \cR^{-1/2} [\cR^{1/2},\ad_x \ad_y].
\end{align}
Then, using the triangle inequality we can bound \eqref{eqn:quadratic:intro} by 
\begin{align}\label{eqn:quadraticwithr:intro}
&\frac{1}{M} \iint\dx\dy\, |\uN(x)| |\uN(y)| |\wN(x-y)| \Big| \lr{\Phi_N, \cR^{-1/2}  \ad_x \ad_y \cR^{1/2}  \PhiNM}\Big| \\
&\qquad+ \frac{1}{M} \iint\dx\dy\, |\uN(x)| |\uN(y)| |\wN(x-y)| \Big|\lr{\Phi_N,  \cR^{-1/2} [\cR^{1/2},\ad_x \ad_y]  \PhiNM}\Big|.\nn 
\end{align}
The key point is that although the first term in \eqref{eqn:quadraticwithr:intro} looks similar to \eqref{eqn:quadratic:intro}, it is much easier to control. Indeed, by the Cauchy--Schwarz inequality, we have
\begin{align}\label{eqn:quadratic:intro-1}
&\frac{1}{M} \iint\dx\dy\, |\uN(x)| |\uN(y)| |\wN(x-y)| \Big| \lr{\Phi_N, \cR^{-1/2}  \ad_x \ad_y \cR^{1/2}  \PhiNM}\Big| \nn\\
&\le \frac{1}{M}\iint\dx\dy\, |\uN(x)| |\uN(y)| |\wN(x-y)| \| a_x a_y \cR^{-1/2}  \Phi_N\| \| \cR^{1/2}  \PhiNM\| \nn\\
&\le \frac{1}{M} \left( \iint\dx\dy\, |\wN(x-y)| \| a_x a_y \cR^{-1/2}  \Phi_N\|^2 \right)^{\frac 1 2} \times \nn\\
&\qquad \qquad \qquad \times \left( \iint \d x \d y |\uN(x)|^2 |\uN(y)|^2 |\wN(x-y)| \right)^{\frac 1 2}  \| \cR^{1/2}   \PhiNM\|. 
\end{align}
Then, by the definition of $\cR$, we can bound 
\begin{align}\label{eqn:quadratic:intro-2}
\iint\dx\dy\, |\wN(x-y)| \| a_x a_y \cR^{-1/2}  \Phi_N\|^2 
&= \langle \Phi_N,  \cR^{-1/2}  \dG_2(|\wN|)  \cR^{-1/2}  \Phi_N\rangle \nn\\
&\le \|\Phi_N\|^2,
\end{align}
without relying on any information on $\Phi_N$. The other factors in \eqref{eqn:quadratic:intro-1} can be bounded efficiently using $\|\wN\|_{L^1}\le C$ and good estimates on $\Phi_{N,M}$ and $\uN$. All this allows to bound \eqref{eqn:quadratic:intro-1} by $C_{t,\eps}N^\eps/\sqrt{M} \|\Phi_N\|$, which appears as the first error term on the right-hand side of \eqref{eq:probB-simplified}.

We still have to bound the second term in \eqref{eqn:quadraticwithr:intro}. This term looks complicated, but in principle, we gain a huge cancelation from the commutator $\left[\cR^\frac12,\ad_x\ad_y\right]$ due to the fact that $\cR$ is a ``local operator''. To make it more transparent, we can use the formula
\begin{equation}
\cR^\frac12=\frac{1}{\pi}\int_0^\infty \frac{1}{\sqrt{s}}\frac{\cR}{\cR+s}\ds
\end{equation}
to write, for  any operator $B$, 
\begin{equation} \label{eq:cR12-com-B}
\left[\cR^\frac12,B\right]
=\frac{1}{\pi}\int_0^\infty \ds \frac{1}{\sqrt{s}} \Big[ \frac{\cR}{\cR+s}, B \Big]  = \frac{1}{\pi}\int_0^\infty \ds \frac{\sqrt{s}}{\cR+s}\left[\cR,B\right]\frac{1}{\cR+s}. 
\end{equation}
In particular, a straightforward computation shows that  
\begin{align}\label{eq:comm-R-aa}
\left[\cR,\ad_x\ad_y\right] &= \left[\dGt(|\wN|),\ad_x\ad_y\right] \\
&=|\wN(x-y)|\ad_x\ad_y+\int\dz\big(|\wN(z-x)|+|\wN(z-y|\big)\ad_x\ad_y\ad_za_z. \nn
\end{align}
Let us take $|\wN(x-y)|\ad_x\ad_y$ from \eqref{eq:comm-R-aa} and insert it in \eqref{eq:cR12-com-B}. The corresponding contribution from the second term in \eqref{eqn:quadraticwithr:intro} can be controlled by 
\begin{equation}
\frac{1}{M}\int_0^\infty\ds \sqrt{s}\iint\dx\dy |\uN(x)| |\uN(y)| |\wN(x-y)|^2 \left\| a_xa_y\frac{\cR^{-\frac12}}{\cR+s}\Phi_N \right\| \left\| \frac{1}{\cR+s}\PhiNM \right\|. 
\end{equation}

The resolvents $(\cR+s)^{-1}$ are important in two respects: one the one hand, they provide sufficient decay in $s$  via the estimate $(\cR+s)^{-1}\leq(1+s)^{-1}$. On the other hand, they compensate for the singular interaction, which is similar to the argument in \eqref{eqn:quadratic:intro-2} although $|\wN|^2$ is way more singular than $|\wN|$.  To combine these two ideas, we use again the Cauchy-Schwarz inequality on $\R^2\times\R^2$, where we split 
\begin{equation} \label{eq:WN-split}
|\wN|^2=|\wN|^{1+\eps/2}|\wN|^{1-\eps/2}
\end{equation}
for $\eps>0$ small and estimate
\begin{equation}
\frac{\cR^{-\frac12}}{\cR+s}\d\Gamma(|\wN|^{2-\eps})\frac{\cR^{-\frac12}}{\cR+s}\leq \frac{ \cR^{1-\eps}} {(\cR+s)^2} \le  \frac{1}{(1+s)^{1+\eps}}. 
\end{equation}
Here, we used that
\begin{equation}
\dGt(|\wN|^{2-\eps})\leq (\dG_2(|\wN|))^{2-\eps}\le \cR^{2-\eps}\,,
\end{equation} 
which relies heavily on the locality of $\cR$, namely $\dGt(|\wN|)$ is the second quantization of a two-body multiplication operator (see Lemma \ref{lem:A^n}). Finally, by calculating the $L^2$-norm of $|\wN|^{1+\eps/2}$, which appears in \eqref{eq:WN-split}, we eventually obtain $C_{t,\eps}N^\eps N^{\beta}/M^{3/2}$,  the second error term on the right-hand side of \eqref{eq:probB-simplified}. This completes our overview of the main ingredients of the proof. 

\bigskip
\noindent
{\bf Organization of the paper.} In Section \ref{sec:Hartree}, we establish uniform-in-$N$ estimates for the Hartree dynamics. The most technical part of the paper is contained in Section \ref{sec:truncated-full} where we prove Proposition \ref{prop:B}. From this, we conclude the main results in Section~\ref{sec:conclusion}.

\section{Uniform estimates for Hartree evolution}\label{sec:Hartree}

In this section, we consider the Hartree evolution $\uN$ in \eqref{hartree}. 
By Assumption~\ref{ass:w}, it is globally well-posed in $H^k$, $k\in\{1,2,\dots\}$ for any fixed $N$ by \cite[Corollary 6.1.2]{Cazenave}. However, it is \textit{a priori} not clear whether $\norm{\uN(t)}_{H^k}$ is bounded uniformly in $N$ for fixed $t\in[0,\Tmax)$. In the following lemma, we prove such uniform bounds for all times prior to the NLS blow-up time $\Tmax$.

\begin{lem}\label{lem:hartree} Let $w$ satisfy Assumption \ref{ass:w}. Let $\varphi_0\in H^4(\R^2)$ and $\Tmax$ as in Lemma~\ref{lem:NLS}. Then, for every $T\in[0,\Tmax)$, there exists a constant $C=C(T,\varphi_0)>0$ such that for all $t\in[0,T]$ and all $N$ sufficiently large, 
\begin{align} \label{eqn:lem:hartree}
\|\uN(t)\|_{L^\infty} \le C\norm{\uN(t)}_{H^2(\R^2)}\leq C, \quad \norm{\partial_t \uN(t)}_{H^2(\R^2)} \le C. 
\end{align}
Moreover, for $\varphi(t)$ the solution of the NLS \eqref{NLS}, it holds that
\begin{equation}\label{eqn:lem:hartree:L2}
\norm{\uN(t)-\varphi(t)}_{L^2(\R^2)}\leq CN^{-\beta}\,.
\end{equation}
\end{lem}

For interactions satisfying the stability condition $\int_{\R^2}|w_-|<a^*$,  \eqref{eqn:lem:hartree} has been shown in \cite[Lemma 4]{NamNap-19}. As explained in  \cite{NamNap-19}, the key point is to get the uniform bound $\norm{\uN(t)}_{H^1(\R^2)}\leq C$, and the rest follows from a rather general argument. In the stability regime considered in \cite{NamNap-19}, this bound follows immediately from  energy conservation and \eqref{eqn:b*}, namely 
\begin{align}
\mathcal{E}[u_N(t)] & = \|\nabla u_N(t)\| + \frac12\iint\dx\dy|\uN(t,x)|^2\wN(x-y)|\uN(t,y)|^2\nn \\
&\ge  \|\nabla u_N(t)\| - \frac12 \|(w_N)_-\|_{L^1} \|u_N(t)\|_{L^4}^4 \nn\\
&\ge \norm{\nabla\uN(t)}^2 \left(1- \frac{\int_{\R^2}|w_-|}{a^*} \right)\label{eqn:hartree:stable}. 
\end{align}

For a general $w$ in our Lemma \ref{lem:hartree}, the estimate \eqref{eqn:hartree:stable} is not available, hence the estimate of $\norm{\uN(t)}_{H^1(\R^2)}$ is more complicated. Instead of studying $u_N$ directly as in \cite{NamNap-19}, we will focus on the difference
\begin{equation}\label{tN}
\tN(t)=\uN(t)-\varphi(t).
\end{equation}
We will bound $\tN(t)$ by a bootstrap argument consisting of two steps:
\begin{enumerate}
\item If $\norm{\tN(t)}\leq \delta=\sqrt{a^*/32\norm{w}_{L^1}}$, then $\norm{\nabla\theta_N(t)}\leq C$. This follows from energy conservation and the \textit{a priori} bound $\norm{\varphi(t)}_{H^1(\R^2)}\leq C$ on $[0,T]$.

\item If $\norm{\tN(s)}\leq\delta$ for all $s\in[0,t]$, then $\norm{\tN(t)}\leq N^{-\beta}\ll\delta$ for sufficiently large $N$. To prove this, we use Step 1 and Gronwall's lemma.
\end{enumerate}
The conclusion thus follows from the continuity of the map $t\mapsto\norm{\tN(t)}$ and the initial condition $\tN(0)=0$. Now let us go to the details.

\begin{proof}[Proof of Lemma \ref{lem:hartree}] 

For simplicity, we consider the solutions $\uN(t)$ and $\varphi(t)$ of the equations \eqref{hartree} and \eqref{NLS} without the phases $\mN(t)$ and $\mu(t)$, respectively. Due to the gauge transformation $\uN\mapsto \exp\{-\i\int_0^t\mu(s)\ds\}\uN$,
this does not change the $N$-dependence of the estimates \eqref{eqn:lem:hartree}. Let $\delta=\sqrt{a^*/(32\norm{w}_{L^1})}$ with $a^*$ as in \eqref{eqn:b*}.

\bigskip

\noindent \textbf{Step 1}. Assume that $\norm{\tN(t)}\leq\delta$ for some $t\in [0,T]$. Using the energy conservation of \eqref{hartree} and the fact $\norm{\wN}_{L^1}=\norm{w}_{L^1}$, we can bound
\begin{align} \label{eqn:pf:hartree:1}
\mathcal{E}[u_N(t)] &=    \mathcal{E}[\varphi_0] = \|\nabla \varphi_0\|^2 + \frac{1}{2} \iint |\varphi_0(x)|^2 w_N(x-y) |\varphi_0(y)|^2 \d x \d y \nn\\
&\le  \|\nabla \varphi_0\|^2 + \frac{1}{2} \|w\|_{L^1} \|\varphi_0\|_{L^4}^4 \le C. 
\end{align}
On the other hand, using \eqref{eqn:b*} and the assumption $\norm{\tN(t)}\leq\delta$,  we can bound
\begin{align*}
\| \theta_N(t) \|_{L^4}^4 \le  \frac{2 \delta^2 }{a^*} \|\nabla \theta_N(t) \|^2 = \frac{1}{16 \|w\|_{L^1}} \|\nabla \theta_N(t) \|^2. 
\end{align*}
Combining with 
$
\frac12\norm{\nabla\tN(t)}^2\leq \norm{\nabla\uN(t)}^2+\norm{\nabla\varphi(t)}^2$, we find that
\begin{align}
\mathcal{E}_N[\uN(t)]
&\geq \norm{\nabla\uN(t)}^2-\frac12\norm{w}_{L^1}\norm{\tN(t)+\varphi(t)}_{L^4}^4\nonumber\\
&\ge \left( \frac{1}{2} \norm{\nabla\tN(t)}^2 -  \norm{\nabla\varphi(t)}^2 \right) - 4 \|w\|_{L^1} (\| \theta_N(t) \|_{L^4}^4 + \| \varphi (t) \|_{L^4}^4 )\nn\\
&\geq\frac14\norm{\nabla\tN(t)}^2-C\,,\label{eqn:pf:hartree:2}
\end{align}
where we used that $\norm{\nabla\varphi(t)}\leq C$ on $[0,T]$. 
Consequently, \eqref{eqn:pf:hartree:1} and \eqref{eqn:pf:hartree:2} imply that
$
\norm{\nabla\tN(t)}^2\leq C$.\\

\noindent\textbf{Step 2.} Let $s\in[0,t]$ and assume that $\norm{\tN(s)}\leq \delta$. Then, dropping the time dependence from the notation, we find that
\begin{align}
\frac12\left|\partial_s\norm{\tN(s)}^2\right|
&=\big|\Im\big\langle\tN,(\wN*|\uN|^2)\tN+(\wN*(|\uN|^2-|\varphi|^2))\varphi\nn\\
& \qquad\qquad+(\wN*|\varphi|^2-b|\varphi|^2)\varphi\big\rangle\big|\nonumber\\
&\leq\norm{\tN} \|\varphi\|_{L^\infty} \left(\norm{\wN*(|\uN|^2-|\varphi|^2)} + \norm{\wN*|\varphi|^2-b|\varphi|^2}\right) \nonumber\\
&=:\norm{\tN} \|\varphi\|_{L^\infty} \left(A_1+A_2\right)\,.\label{eqn:pf:hartree:3}
\end{align}
On the right hand side of \eqref{eqn:pf:hartree:3}, we have  $\norm{\tN(s)}\leq \delta$ by our assumption, and  $\norm{\varphi(s)}_{L^\infty}\leq C$ since $\varphi(s)\in H^2(\R^2)$ by \cite[Theorems 5.3.1 and 5.4.1]{Cazenave} and Sobolev's embedding $H^2(\R^2)\subset L^\infty(\R^2)$ (\cite[Theorem 8.8~(iii)]{LieLos-01}). 

\bigskip

\noindent{\it Estimate of $A_1$}. Using 
\begin{equation}
\Big| |\uN|^2-|\varphi|^2 \Big|=\Big|(|\uN|-|\varphi|)(|\uN|+|\varphi|) \Big|\leq |\tN|^2+2|\varphi||\tN|\,,
\end{equation}
we can bound 
\begin{equation}\label{eq:uN-A1}
A_1\leq \|w_N\|_{L^1} (\|\theta_N\|_{L^4}^2 + 2\|\varphi\|_{L^\infty} \|\theta_N\|) \le C\norm{\tN}\,.
\end{equation} 
In the last estimate, we used \eqref{eqn:b*} and the bound $
\norm{\nabla\tN(t)}^2\leq C$ from Step 1.

\bigskip

\noindent{\it Estimate of $A_2$}. Observing that $b=\hat{w}_N(0)$ and $\hat{w}_N(\xi)=\hat{w}(\xi/N^\beta)$, Plancherel's theorem yields
\begin{equation}\label{eq:uN-A2}
A_2\leq \left\|\frac{\hat{w}(\cdot/N^\beta)-\hat{w}(0)}{|\cdot|}\right\|_{L^\infty}\norm{|\cdot|\hat{|\varphi|^2}}_{L^2}\leq CN^{-\beta}\,.
\end{equation}
Here, we used $\|\nabla \varphi\|\le C$ and the fact that $\hat w$ is Lipschitz.

In summary, inserting \eqref{eq:uN-A1} and \eqref{eq:uN-A2} in \eqref{eqn:pf:hartree:3}, we arrive at 
\begin{equation}
\partial_s\norm{\tN(s)}^2\leq C\left(\norm{\tN(s)}^2+N^{-2\beta}\right).
\end{equation}
Consequently, we obtain $\norm{\tN(t)}\leq C N^{-\beta}$ by Gronwall's lemma since $\tN(0)=0$. 

\bigskip

\noindent\textbf{Conclusion}.
Define 
\begin{equation}
t_N^\mathrm{max}=\sup\{t\in[0,T]:\norm{\tN(t)}\leq\delta\}\,.
\end{equation}
Assume that $t_N^\mathrm{max}< T$.
By \cite[Theorem 4.10.1]{Cazenave}, the map $[0,T]\ni t\mapsto\norm{\tN(t)}$ is continuous, hence  $\norm{\tN(s)}\leq\delta$ for $s\in[0,t_N^\mathrm{max}]$. By Step 2, this implies that
\begin{equation}
\norm{\tN(t_N^\mathrm{max})}\leq CN^{-\beta}<\delta
\end{equation}
for sufficiently large $N$, which contradicts $t_N^\mathrm{max}<T$. Hence,  $t_N^\mathrm{max}\geq T$, and consequently 
$$\norm{\tN(t)}\leq\delta,\quad \forall t\in[0,T].$$
By Step 1, we get $
\norm{\nabla\tN(t)}^2\leq C$. Therefore,
\begin{equation}
\| u_N(t)\|_{H^1}\le \|\tN(t) \|_{H^1} + \|\varphi(t) \|_{H^1}\le C, \quad \forall t\in[0,T]\,.
\end{equation}
The remaining estimates in \eqref{eqn:lem:hartree} can be deduced from the $H^1$-bound as in \cite[Lemma 4]{NamNap-19}, using Duhamel's formula. The bound \eqref{eqn:lem:hartree:L2} also follows from the above argument, where the error term $N^{-\beta}$ comes from \eqref{eq:uN-A2}. 
\end{proof}

\section{From the truncated to the full dynamics}\label{sec:truncated-full}

In this section we prove Proposition \ref{prop:B}. As explained in Section \ref{sec:stratefy-3}, the key step is to prove the propagation bound \eqref{eq:probB-simplified}. 
We use \eqref{eq:quadratic-generator} and \eqref{eqn:GN} to decompose 
\begin{align}\label{eq:lemB-dec}
\left|\frac{\d}{\d t} \lr{\PhiN(t),f_M^2\PhiNM(t)}  \right| \le \sum_{j=1}^3 \left|\lr{\PhiN, \Big[(\bG_j + \bG_j^*) ,\fM^2\Big]\PhiNM}\right| 
\end{align}
with $\bG_j$ given in \eqref{eqn:GN}. In the next subsections, we will handle the cases $j=1,2,3$ separately, and then conclude \eqref{eq:probB-simplified} as well as Proposition \ref{prop:B}.

As a preparation, let us collect here two auxiliary estimates which will be used repeatedly in this section. The first one is a simple Sobolev-type estimate.

\begin{lem}\label{lem:sobolev} Let $W \in L^{s}(\R^2)$ with $s\in(1,2]$ and denote by $W(x-y)$ the corresponding two-body multiplication operator. Then 
\begin{align}\label{eq:sobolevW}
\dG_2(|W(x-y)|) \le C_s \|W\|_{L^{s}(\R^2)} \cN \dG_1(1-\Delta)
\end{align}
as operators on $\cF$. 
\end{lem}

In particular, Assumption \ref{ass:w} guarantees that $w\in L^{1+\eps}(\R^2)$ for every $\eps>0$, hence  Lemma \ref{lem:sobolev} implies that
\begin{align}\label{lem:sobolev:w}
\dG_2(|w_N(x-y)|) \le C_\eps N^\eps \cN \dG_1(1-\Delta).
\end{align}
Here we used the fact that 
\begin{equation}
\int_{\R^2} |w_N(x)|^\alpha \d x = N^{2\beta (\alpha-1)} \int |w(x)|^\alpha \d x, \quad \forall \alpha>0. 
\end{equation}

\begin{proof} 
Using Sobolev's embedding $L^{s'}(\R^2)\supset H^1(\R^2)$ with $1/s' + 1/s=1$ \cite[Theorem 8.8]{LieLos-01}, we have
\begin{align}
\iint \d x \d y |W(x-y)| |f(x,y)|^2 &\le C_s \int \d x \| W(x-\cdot)\|_{L^s(\R^2)}  \|f(x,\cdot)\|_{H^1(\R^2)}^2 \nn\\
&= C_s  \|W\|_{L^s(\R^2)} \langle f, (1-\Delta_y) f\rangle_{L^2(\R^2\times \R^2)}
\end{align}
for all $f\in H^1(\R^2\times \R^2)$. Therefore, we have the 2-body inequality 
\begin{align}
|W(x-y)| \le C_s \|W\|_{L^s(\R^2)} (1-\Delta_y)
\end{align}
for each $y\in\R^2$, which implies the second-quantized form \eqref{eq:sobolevW}. 
\end{proof}

The second estimate is concerned with the second quantization of a two-body multiplication operator:

\begin{lem}\label{lem:A^n}
Let $A\geq0$ be a multiplication operator on $\fH^2$ such that $A(x,y)=A(y,x)$ and let $s\in [1,\infty)$. Then
\begin{equation}
\dGt(A^s)\leq \left[\dGt(A)\right]^s\,.
\end{equation}
\end{lem}
\begin{proof} On every $k$-particle sector $\fH^k$, $k\ge 2$, we have
\begin{align}
\dGt(A^s) = \sum_{1\le i<j \le k} A_{ij}^s \le \Big( \sum_{1\le i<j \le k} A_{ij} \Big)^s= \left[\dGt(A)\right]^s.
\end{align}
This concludes the proof since $\dGt(A)$ preserves the particle number. 
\end{proof}

\subsection{Estimate of the linear terms}\label{subsec:linear} 
We consider first the linear terms in \eqref{eq:lemB-dec}. 

\begin{lem} \label{lem:linear} For every $t\in[0,\Tmax)$ and $\eps>0$,  we have
\begin{align} 
\left| \lr{\PhiN,\left[ (\bG_1 + \bG_1^*) ,\fM^2\right]\PhiNM} \right| \le  C_{t,\eps} \frac{N^\eps}{\sqrt{N}}\,. \label{eq:B-linear-1}
\end{align}
\end{lem}

\begin{proof} 
From the definition of $\bG_1$ in \eqref{eqn:GN}, we obtain
\begin{align}\label{eq:bG1-0}
\Big| \langle \PhiN, [\bG_1,\fM^2] \Phi_{N,M} \rangle \Big| = 2 \Big| \Big\langle \PhiN, \ad \big(q(t)(\wN*|\uN|^2)\uN\big) \omega_1 \Phi_{N,M}\Big\rangle\Big|
\end{align}
with  
\begin{equation}
\omega_1 =  \frac{\Number\sqrt{N-\Number}}{N-1}  \left( f^2\left( \frac{\cN}{M} \right) - f^2\left( \frac{\cN+1}{M} \right)   \right)\,,
\end{equation}
where we used that $g(\cN)\ad_x= \ad_x g(\cN+1)$. For $f$ as in \eqref{eq:fM}, we have 
\begin{align} \label{eq:omega_1-bound} 
|\omega_1|\le \frac{C\cN}{M\sqrt{N}} \id^{\le M}
\end{align} 
in the sense of operators on $\FNp(t)$. 
We will use the simple bound 
\begin{align} \label{eq:aa*-a*a} 
a(v) \ad(v) = \ad(v) a(v) + \|v\|^2   \le (\cN+1) \|v\|^2\,,
\end{align} 
where $v= q(t)(\wN*|\uN|^2)\uN$ satisfies
\begin{equation}
\|v\|_{L^2} \le \| (\wN*|\uN|^2)\uN \|_{L^2}\le \| \wN\|_{L^1} \|\uN\|_{L^2}^2 \|\uN\|_{L^\infty} \le C_t 
\end{equation}
by Lemma \ref{lem:hartree}. Therefore, by the Cauchy--Schwarz inequality, we deduce from \eqref{eq:bG1-0} and \eqref{eq:aa*-a*a} that 
\begin{align}\label{eq:bG1-1}
\Big| \langle \PhiN, [\bG_1,\fM^2] \Phi_{N,M} \rangle \Big|  &\le 2 \|\PhiN\| \Big\| \ad \big(q(t)(\wN*|\uN|^2)\uN\big) (\cN+1)^{-1/2} \Big\|_{\rm op} \times\nn\\
&\qquad \qquad \qquad \times \Big\| (\cN+1)^{1/2} \omega_1 \Phi_{N,M} \Big\| \nn\\
&\le C_t  \Big\| \frac{(\cN+1)^{3/2}}{M\sqrt{N}} \id^{\le M} \Phi_{N,M} \Big\| \le C_{t,\eps} \frac{N^\eps}{\sqrt{N}} . 
\end{align}
In the last estimate, we used that
\begin{align} \label{eq:cN-PhiNM}
\|\cN^{1/2}\Phi_{N,M}\|^2 = \langle \Phi_{N,M}, \cN \Phi_{N,M}\rangle \le C_{t,\eps}N^\eps\,,
\end{align}
which follows from the kinetic estimate \eqref{eq:nam:11} in Lemma \ref{lem:nam}. Similarly, we also get
\begin{align}\label{eq:bG1t-1}
\Big| \langle \PhiN, [\bG_1^*,\fM^2] \Phi_{N,M} \rangle \Big|   \le C_{t,\eps} \frac{N^\eps}{\sqrt{N}}  
\end{align}
since 
\begin{align}\label{eq:bG1t-0}
\Big| \langle \PhiN, [\bG_1^*,\fM^2] \Phi_{N,M} \rangle \Big| = 2 \Big| \Big\langle \PhiN, a \big(q(t)(\wN*|\uN|^2)\uN\big) \tilde{\omega}_1 \Phi_{N,M}\Big\rangle\Big| \,,
\end{align}
where
\begin{equation}
\tilde{\omega}_1 =  \frac{\Number\sqrt{N-\Number+1}}{N-1}  \left( f^2\left( \frac{\cN}{M} \right) - f^2\left( \frac{\cN-1}{M} \right)   \right)
\end{equation}
as operator on $\FNp$.
From \eqref{eq:bG1-1} and \eqref{eq:bG1t-1}, we obtain \eqref{eq:B-linear-1}. 
\end{proof}

\subsection{Estimate of the quadratic terms}\label{subsec:quadratic} We turn to  the quadratic terms in \eqref{eq:lemB-dec}. 

\begin{lem} \label{lem:quadratic} For every $t\in[0,\Tmax)$ and $\eps>0$, we have
\begin{align} 
\left| \lr{\PhiN,\left[\bGt,\fM^2\right]\PhiNM} \right| &\le  C_{t,\eps} \left( \frac{1}{\sqrt{M}} + \frac{N^\beta}{M^{\frac 3 2}} \right) N^{\eps} , \label{eq:B-quadratic-1}\\
\left| \lr{\PhiN,\left[\bGt^*,\fM^2\right]\PhiNM} \right| &\le  C_{t,\eps}   \frac{N^\eps}{M}. \label{eq:B-quadratic-2}
\end{align}
\end{lem}

\begin{proof} The bound \eqref{eq:B-quadratic-1} is one of the most difficult estimates in this section. We use the strategy explained in Section \ref{sec:stratefy-3}.

\bigskip
\noindent
{\bf Step 1.} Let us abbreviate
\begin{equation}
\omega_2=\frac{\sqrt{(N-\Number)(N-\Number-1)}}{N-1} \left( f^2\Big(\frac{\Number}{M}\Big) - f^2 \Big(\frac{\Number+2}{M}\Big) \right)
\end{equation}
as operator on $\FNp$.
For $N\geq 2$, we have 
\begin{equation}\label{eq:omega_2-bound}
|\omega_2|\leq \frac{C}{M} \id^{>M/2}\,.
\end{equation}
We also observe that in the relevant estimate for $\bG_2$,  $K_2=q\otimes q\tilde{K}_2$ in \eqref{K_1:K_2}  can be replaced by $\tilde{K}_2$ as for any $\chi,\chi'\in\Fp(t)$ we have
\begin{align}
&\lr{\chi,\iint\dx\dy K_2(x,y)\uN(x)\uN(y)\ad_x\ad_y\chi'} \nn\\
&\qquad \qquad \qquad = \lr{\chi,\iint\dx\dy\tilde{K}_2(x,y)\uN(x)\uN(y)\ad_x\ad_y\chi'}\label{eqn:drop:q}.
\end{align} 
Hence, 
 we can write 
\begin{align} \label{eq:B-G2-exact}
\lr{\PhiN,\left[\bG_2,\fM^2\right]\PhiNM} =  \iint 
\dx\dy \wN(x-y) \uN(x) \uN(y) \Big\langle \PhiN, \ad_x \ad_y \omega_2 \Phi_{N,M} \Big\rangle. 
\end{align}
By decomposing 
\begin{align}
\ad_x \ad_y = \cR^{-1/2}  \ad_x \ad_y \cR^{1/2} +  \cR^{-1/2} [\cR^{1/2},\ad_x \ad_y]
\end{align}
with $\cR=\dGt(|w_N|)+1$ as in \eqref{def:R}
and using the triangle inequality, we find that 
\begin{align}\label{eq:G_2-dec}
\left|\lr{\PhiN,\left[\bGt,\fM^2\right]\PhiNM}\right| \le \mathcal{E}_1+ \mathcal{E}_2
\end{align}
where
\begin{align}
\mathcal{E}_1 &= \iint\dx\dy|\wN(x-y)||\uN(x)| |\uN(y)|\left|\Big\langle  \PhiN, \cR^{-1/2}  \ad_x \ad_y  \cR^{1/2}  \omega_2 \Phi_{N,M} \Big\rangle \right|,\\
\mathcal{E}_2 &= \iint\dx\dy|\wN(x-y)||\uN(x)| |\uN(y)| \left|\lr{\PhiN, \cR^{-\frac12}\left[\cR^\frac12,\ad_x\ad_y\right]\omega_2\PhiNM}\right|. 
\end{align}

\bigskip
\noindent
{\bf Step 2.} Now let us estimate $\cE_1$. By the Cauchy-Schwarz inequality, 
\begin{align}
\mathcal{E}_1  &\le  \iint\dx\dy|\wN(x-y)||\uN(x)| |\uN(y)|  \| a_x a_y  \cR^{-1/2} \PhiN \| \|    \cR^{1/2}  \omega_2 \Phi_{N,M} \| \nonumber\\
&\leq \left(\iint \d x \dy |\wN(x-y) |\uN(x)|^2 |\uN(y)|^2\right)^\frac12 \times \nonumber \\
&\quad\times \left( \iint \d x \d y |w_N(x-y)| \|a_x a_y \cR^{-1/2} \PhiN\|^2 \right)^{\frac 1 2} \| \cR^{1/2}  \omega_2 \Phi_{N,M} \|.\label{cE1-0}
\end{align}
From Lemma \ref{lem:hartree}, we can bound 
\begin{align}\label{cE1-1}
 \iint \d x \dy |\wN(x-y) |\uN(x)|^2 |\uN(y)|^2 &\le \|\uN\|_{L^\infty}^2 \|u_N\|_{L^2}^2 \|w_N\|_{L^1} \le   C_t. 
\end{align}
Moreover,  \eqref{eqn:quadratic:intro-2} yields
\begin{align}
 \iint \d x \d y |w_N(x-y)| \|a_x a_y \cR^{-1/2} \PhiN\|^2  
 &\le \|\PhiN\|^2 \le 1.\label{cE1-2}
\end{align}
From Lemma \ref{lem:sobolev} and the kinetic estimate in  Lemma \ref{lem:nam}, we get
\begin{align} \label{eq:cR-PhiNM}
\Big \langle  \Phi_{N,M} ,  \dGt(|w_N|)  \Phi_{N,M} \Big \rangle \le C_{\eps}N^\varepsilon \Big\langle \Phi_{N,M}, M \dGo(1-\Delta) \Phi_{N,M} \Big \rangle \le C_{t,\eps} M N^{2\eps}.
\end{align}
Combining this with \eqref{eq:omega_2-bound} and the fact that $\cR$ commutes with $\omega_2$, we find that 
\begin{align}\label{cE1-3}
\| \cR^{1/2}  \omega_2 \Phi_{N,M} \|^2 &\le \frac{C}{M^2} \Big \langle   \Phi_{N,M}, \cR  \Phi_{N,M} \Big \rangle \le \frac{C_{t,\eps} N^{\eps}}{M}.
\end{align}
Inserting \eqref{cE1-1}, \eqref{cE1-2} and \eqref{cE1-3} in \eqref{cE1-0}, we conclude that
\begin{align} \label{eq:cE1-final}
\mathcal{E}_1 \le \frac{C_{t,\varepsilon} N^\varepsilon}{\sqrt{M}}\norm{\PhiN}
\end{align}
for every constant $\eps>0$. 

\bigskip
\noindent
{\bf Step 3.} We turn to estimate the second term $\cE_2$, which is more involved. Using \eqref{eq:comm-R-aa} and \eqref{eq:cR12-com-B}, we get
\begin{align}\label{eqn:commutator:R12}
\left[\cR^\frac12,\ad_x\ad_y\right]&=\frac{1}{\pi}\int_0^\infty\ds\,\frac{ \sqrt{s}}{\cR+s} |\wN(x-y)|\ad_x\ad_y \frac{1}{\cR+s} \\ 
&+\frac{1}{\pi}\int_0^\infty\ds \int \d z \,\frac{ \sqrt{s}}{\cR+s} \big(|\wN(x-z)|+|\wN(y-z)|\big)\ad_x\ad_y\ad_za_z \frac{1}{\cR+s}.\nn
\end{align}
This allows us to decompose 
\begin{align} \label{eq:cE2-0}
\cE_2 &=\iint\dx\dy|\wN(x-y)||\uN(x)| |\uN(y)| \left|\lr{\PhiN, \cR^{-\frac12}\left[\cR^\frac12,\ad_x\ad_y\right]\omega_2\PhiNM}\right| \nn\\
&\le C (\cE_{2,1} + \cE_{2,2}) 
\end{align}
where
\begin{align} 
\cE_{2,1} &= \int_0^\infty  \d s \sqrt{s}   \iint\dx\dy|\wN(x-y)|^2 |\uN(x)| |\uN(y)| \times\nn \\
&\qquad \qquad \qquad \qquad\times \left|\lr{\PhiN, \frac{\cR^{-\frac12}}{\cR+s}\ad_x\ad_y \frac{1}{\cR+s}\omega_2\PhiNM}\right|,\\
\cE_{2,2} &=\int_0^\infty \d s \sqrt{s}   \iiint\dx\dy \d z |\wN(x-y)| |w_N(x-z)| |\uN(x)| |\uN(y)| \times \nn\\
&\qquad \qquad \qquad  \times \left|\lr{\PhiN, \frac{\cR^{-\frac12}}{\cR+s}\ad_x\ad_y \ad_z a_z \frac{1}{\cR+s}\omega_2\PhiNM}\right|. 
\end{align}

\medskip
\noindent
{\bf Estimate of $\cE_{2,1}$.}
By the Cauchy--Schwarz inequality, we find for any constant $\varepsilon\in(0,1)$ that
\begin{align}  \label{eq:cE21-0}
 \cE_{2,1} &\le \int_0^\infty  \d s \sqrt{s}   \iint\dx\dy|\wN(x-y)|^2 |\uN(x)| |\uN(y)| \times \nn\\
  &\qquad \qquad \qquad \times  \left \| a_x a_y  \frac{\cR^{-\frac12}}{\cR+s} \PhiN \right\|  \left\| \frac{1}{\cR+s}\omega_2\PhiNM \right\| \nn\\
 &\le\int_0^\infty  \d s \sqrt{s}   \left(  \iint\dx\dy|\wN(x-y)|^{2+\eps}  |\uN(x)|^2 |\uN(y)|^2\right)^{\frac 1 2} \times \nn\\
 &\quad \times \left( \iint \d x \d y |w_N(x-y)|^{2-\eps} \left \| a_x a_y  \frac{\cR^{-\frac12}}{\cR+s} \PhiN \right\|^2  \right)^{\frac 1 2}  \left\| \frac{1}{\cR+s}\omega_2\PhiNM \right\| .
\end{align}
By Lemma \ref{lem:hartree}, we obtain
\begin{align}  \label{eq:cE21-1}
\iint\dx\dy|\wN(x-y)|^{2+\eps} |\uN(x)|^2 |\uN(y)|^2 &\le \|u_N\|_{L^\infty}^2 \|u_N\|_{L^2}^2 \||w_N|^{2+\eps} \|_{L^1} \nn\\
&\le C_{t} N^{2\beta (1+\eps)}. 
\end{align}
Moreover, using that
\begin{align} \label{eq:dG2-wN-2eps}
 \iint \d x \d y |w_N(x-y)|^{2-\eps}  \ad_x \ad_y a_x a_y = \dG_2(|w_N|^{2-\eps})\le (\dG_2(|w_N|))^{2-\eps} \le  \cR^{2-\eps}
\end{align}
by Lemma \ref{lem:A^n},
we can bound
\begin{align}  \label{eq:cE21-2}
& \iint \d x \d y |w_N(x-y)|^{2-\eps} \left \| a_x a_y  \frac{\cR^{-\frac12}}{\cR+s} \PhiN \right\|^2 \nn = \left \langle \frac{\cR^{-\frac12}}{\cR+s} \PhiN , \dG_2(|w_N|^{2-\eps}) \frac{\cR^{-\frac12}}{\cR+s} \PhiN  \right\rangle \nn\\
 & \le \left \langle \PhiN, \frac{\cR^{1-\eps}}{(\cR+s)^2} \PhiN \right\rangle \le \frac{1}{(1+s)^{1+\eps}}. 
\end{align}
In the last estimate we used that $\cR\ge 1$. Moreover, using again the fact that $\cR$ commutes with $\omega_2$, we find with \eqref{eq:omega_2-bound} that
\begin{align}  \label{eq:cE21-3}
 \left\| \frac{1}{\cR+s}\omega_2\PhiNM \right\|^2 &\le \frac{C}{M^2 (1+s)^2} \langle \PhiNM, \id^{>M/2} \PhiNM  \rangle \nn\\ 
 &\le  \frac{C}{M^2 (1+s)^2} \left\langle \PhiNM, \frac{2\Number}{M}\PhiNM \right \rangle \le  \frac{C_{t,\eps} }{M^3 (1+s)^2} N^\eps. 
\end{align}
Here in the last estimate, we used the kinetic bound in Lemma \ref{lem:nam}. Inserting \eqref{eq:cE21-1}, \eqref{eq:cE21-2} and \eqref{eq:cE21-3} in \eqref{eq:cE21-0} we find that, for every constant $\eps\in(0,1)$, 
\begin{align}\label{eq:cE21-final}
 \cE_{2,1} &\le C_{t,\eps} \int_0^\infty  \d s \sqrt{s}  \sqrt{N^{2\beta(1+\eps)}}\sqrt{ \frac{1}{(1+s)^{1+\eps}} } \sqrt{ \frac{1}{M^3 (1+s)^2} N^\eps } \nn\\
 &\le C_{t,\eps} \frac{N^{(1+\eps)\beta+\eps/2}}{M^{\frac 3 2}} \int_0^\infty \frac{ \d s}{(1+s)^{1+\eps/2}} \le  C_{t,\eps} \frac{N^{(1+\eps)\beta+\eps/2}}{{M^{\frac 3 2}}} . 
 \end{align}

\medskip
\noindent
{\bf Estimate of $\cE_{2,2}$.}
Similarly, for every constant $\eps>0$ small, by the Cauchy--Schwarz inequality, 
\begin{align}  \label{eq:cE22-0}
&\cE_{2,2} =\int_0^\infty \d s \sqrt{s}   \iiint\dx\dy \d z |\wN(x-y)| |w_N(x-z)| |\uN(x)| |\uN(y)| \times \nn\\
&\qquad \qquad \qquad \times \left|\lr{ \PhiN, \frac{\cR^{-\frac12} (\cN+1)^{-1/2}}{\cR+s}  \ad_x\ad_y \ad_z a_z  \frac{(\cN+3)^{1/2}}{\cR+s}\omega_2\PhiNM}\right| \nn\\
&\le \|\uN\|_{L^\infty}^2 \int_0^\infty \d s \sqrt{s}   \iiint\dx\dy \d z |\wN(x-y)| |w_N(x-z)|  \times \nn\\
&\qquad\qquad\qquad \times \left\|  a_x a_y a_z \frac{\cR^{-\frac12}(\cN+1)^{-1/2}}{\cR+s} \PhiN \right\| \left\| a_z \frac{(\cN+3)^{1/2}}{\cR+s}\omega_2\PhiNM\right\| \nn\\
&\le  C_t  \int_0^\infty \d s \sqrt{s}   \left( \iiint \d x \d y \d z |w_N(x-y)|^{2-\eps}  \left\| a_x a_y a_z \frac{\cR^{-\frac12}(\cN+1)^{-1/2}}{\cR+s} \PhiN \right\|^2  \right)^{ \frac 1 2} \nn\\
&\quad  \times  \left( \int\dz \left\| a_z \frac{(\cN+3)^{1/2}}{\cR+s}\omega_2\PhiNM\right\| ^2\int\dx |w_N(x-z)|^{2}\int\dy |w_N(x-y)|^\eps  \right)^{\frac 1 2}.
\end{align}
In the last estimate, we used the uniform bound $\|\uN\|_{L^\infty}\le C_t$ from Lemma \ref{lem:hartree}. Using again \eqref{eq:dG2-wN-2eps} and $\cR\ge 1$ we find that  
\begin{align} 
&\iiint \d x \d y \d z |w_N(x-y)|^{2-\eps}  \left\| a_x a_y a_z \frac{\cR^{-\frac12}(\cN+1)^{-1/2}}{\cR+s} \PhiN \right\|^2 \nonumber\\
&= \left \langle  \frac{\cR^{-\frac12}(\cN+1)^{-1/2}}{\cR+s} \PhiN ,  \dG_2 (|w_N|^{2-\eps}) \cN \frac{\cR^{-\frac12}(\cN+1)^{-1/2}}{\cR+s} \PhiN   \right \rangle\nonumber\\
&\le \left \langle \PhiN ,  \frac{\cR^{1-\eps}}{(\cR+s)^2}  \PhiN   \right \rangle \le \frac{1}{(1+s)^{1+\eps}}. 
\end{align}
Since $w$ is bounded and compactly supported, we get
\begin{align}
\int  \d x |w_N(x-z)|^{2}\int\d y  |w_N(x-y)|^{\eps} \le C N^{2\beta \eps}.
\end{align}
Moreover, using \eqref{eq:omega_2-bound} together with 
$\cN^2 \le M \dG_1(1-\Delta)$ 
on $\Fock^{\le M}$ and Lemma \ref{lem:nam}, we have
\begin{align}
\int \d z \left\| a_z \frac{(\cN+3)^{1/2}}{\cR+s}\omega_2\PhiNM\right\|^2 &= \frac{C}{M^2}\left\langle \PhiNM, \frac{\cN(\cN+3)}{(\cR+s)^2} \PhiNM \right\rangle \\
&\leq \frac{C}{M^2(1+s)^2}\lr{\PhiNM,\Number^2\PhiNM}\le  \frac{C_{t,\eps} N^\eps}{M (1+s)^2}\,.  
\end{align}
Therefore, we deduce from \eqref{eq:cE22-0} that
\begin{align}  \label{eq:cE22-final}
\cE_{2,2} \le C_{t,\eps} \int_0^\infty \d s \sqrt{s}  \sqrt{\frac{1}{(1+s)^{1+\eps}} } \sqrt{\frac{N^\eps}{M (1+s)^2} N^{2\beta \eps}} 
\le C_{t,\eps} \frac{N^{(\beta+1/2)\eps}}{\sqrt{M}} .   
\end{align}
Putting \eqref{eq:cE21-final} and \eqref{eq:cE22-final} together, we conclude from \eqref{eq:cE2-0} that
\begin{align}  \label{eq:cE2-final}
\cE_{2} \le  C_{t,\eps} \left( \frac{N^{(1+\eps)\beta+\eps/2}}{{M^{\frac 3 2}}} + 
 \frac{N^{(\beta+1/2)\eps}}{\sqrt{M}} \right) .
\end{align}

\medskip
\noindent
{\bf Conclusion of \eqref{eq:B-quadratic-1}:} Inserting \eqref{eq:cE1-final} and \eqref{eq:cE2-final} in \eqref{eq:G_2-dec}, we obtain \eqref{eq:B-quadratic-1}. 

\bigskip
\noindent
{\bf Step 4.} It remains to prove \eqref{eq:B-quadratic-2}. Similarly to \eqref{eq:B-G2-exact}, we can write
\begin{align} \label{eq:B-G2*-exact}
\lr{\PhiN,\left[\bG_2^*,\fM^2\right]\PhiNM} =  \iint 
\dx\dy \wN(x-y) \overline{\uN(x) \uN(y)} \Big\langle \PhiN, a_x a_y  \tilde{\omega}_2 \Phi_{N,M} \Big\rangle
\end{align}
with 
\begin{equation}
\tilde\omega_2=\frac{\sqrt{(N-\Number + 2)(N-\Number+1)}}{N-1} \left( f^2\Big(\frac{\Number-2}{M}\Big) - f^2 \Big(\frac{\Number}{M}\Big) \right) 
\end{equation}
as operator on $\FNp$.
This term is much easier to estimate than \eqref{eq:B-G2-exact} since now two annihilators hit $\Phi_{N,M}$. To be precise, we have
\begin{equation}\label{eq:t-omega_2-bound}
|\tilde\omega_2|\leq \frac{C}{M}
\end{equation} 
similarly to \eqref{eq:omega_2-bound}.
Therefore, by the Cauchy--Schwarz inequality,
\begin{align} \label{eq:B-G2*-final}
\Big| \lr{\PhiN,\left[\bG_2^*,\fM^2\right]\PhiNM} \Big| &\le \iint \dx\dy |\wN(x-y)| |\uN(x)|  |\uN(y) | \| \PhiN \| \|a_x a_y  \tilde{\omega}_2 \Phi_{N,M} \| \nn\\
&\le \|\PhiN\| \left( \iint \dx\dy |\wN(x-y)| |\uN(x)|^2  |\uN(y) |^2 \right)^\frac 1 2 \times \nn\\
&\qquad \qquad \qquad \times  \left(  \iint \dx\dy |\wN(x-y)|  \|a_x a_y  \tilde{\omega}_2 \Phi_{N,M}\|^2 \right)^\frac 1 2 \nn\\
&\le C_t  \left\langle \Phi_{N,M}, \dG_2(|w_N|) |\tilde{\omega}_2|^2 \Phi_{N,M} \right\rangle^{1/2}\nn\\
&\le \frac{C_{t,\eps}}{M^2}  \left\langle \Phi_{N,M}, N^\eps \cN \dG_1(1-\Delta)  \Phi_{N,M} \right\rangle^{1/2}\le C_{t,\eps} \frac{N^\eps}{M} . 
\end{align}
Here we used again \eqref{cE1-1}, Lemma \ref{lem:sobolev} and the kinetic estimate in  Lemma \ref{lem:nam}. Thus,  \eqref{eq:B-quadratic-2} holds true. This completes the proof of Lemma \ref{lem:quadratic}. 
 \end{proof}

\subsection{Estimate of the cubic terms}\label{subsec:cubic}

Concerning the cubic terms in \eqref{eq:lemB-dec}, we have the following bounds: 

\begin{lem} \label{lem:cubic} Let $\PhiN \in \cF_{\bot}(t)$, $t\in[0,\Tmax)$ and $\eps>0$. Then 
\begin{align} 
\left| \lr{\PhiN,\left[\bG_3,\fM^2\right]\PhiNM} \right| &\le  C_{t,\eps} \left( \frac{1}{\sqrt{N}} + \frac{N^\beta}{M\sqrt{N}} \right) N^{\eps} , \label{eq:B-cubic-1}\\
\left| \lr{\PhiN,\left[\bG_3^*,\fM^2\right]\PhiNM} \right| &\le  C_{t,\eps}   \frac{N^\eps}{\sqrt{N}} . \label{eq:B-cubic-2}
\end{align}
\end{lem}

\begin{proof} Again, the bound \eqref{eq:B-cubic-1} is much more difficult than \eqref{eq:B-cubic-2}. We will proceed similarly to the quadratic terms.

\bigskip
\noindent
{\bf Step 1.} Analogously to \eqref{eq:omega_2-bound}, we denote 
\begin{equation}
\omega_3=\sqrt{1-\frac{\Number-1}{N-1}} \left( f^2\Big(\frac{\Number}{M}\Big) - f^2\Big(\frac{\Number+1}{M}\Big) \right) 
\end{equation}
as operator on $\FNp$, which satisfies 
\begin{equation}\label{eq:omega_3-bound}
|\omega_3|\leq \frac{C}{M} \id^{\le M}.
\end{equation} 
Moreover, similarly to \eqref{eq:B-G2-exact} we can write
\begin{align} \label{eq:B-GB3-0}
\lr{\PhiN,\left[\bG_3,\fM^2\right]\PhiNM} = \frac{1}{\sqrt{N}} \iint \dx\dy \wN(x-y) \uN(x)\Big\langle \PhiN, \ad_x \ad_y a_y \omega_3 \Phi_{N,M} \Big\rangle. 
\end{align}
By decomposing 
\begin{equation}
\ad_x \ad_y a_y = \cR^{-1/2}\ad_x \ad_y a_y \cR^{1/2} + \cR^{-1/2}[\ad_x \ad_y a_y, \cR^{1/2}] \,,
\end{equation}
we obtain
\begin{align} \label{eq:G3-dec}
\left| \lr{\PhiN,\left[\bG_3,\fM^2\right]\PhiNM} \right| \le \mathcal{E}_3+ \mathcal{E}_4  \,,
\end{align}
where
\begin{align} 
\mathcal{E}_3&= \frac{1}{\sqrt{N} } \iint \dx\dy |\wN(x-y)| |\uN(x)| \Big|\Big\langle \PhiN, \cR^{-1/2}\ad_x \ad_y a_y \cR^{1/2} \omega_3 \Phi_{N,M} \Big\rangle \Big|,\\
\mathcal{E}_4& = \frac{1}{\sqrt{N} } \iint \dx\dy |\wN(x-y)| |\uN(x)| \Big|\Big\langle \PhiN, \cR^{-1/2}[\ad_x \ad_y a_y, \cR^{1/2}] \omega_3 \Phi_{N,M} \Big\rangle \Big|.
\end{align}

\bigskip
\noindent
{\bf Step 2.} Let us first estimate $\cE_3$.
By the Cauchy--Schwarz inequality, 
\begin{align}\label{eq:E3-0}
\mathcal{E}_3& \le \frac{1}{\sqrt{N} } \iint \dx\dy |\wN(x-y)| |\uN(x)| \| a_x a_y \cR^{-1/2} \PhiN \|  \| a_y \cR^{1/2} \omega_3 \Phi_{N,M} \| \nn\\
&\le  \frac{\|\uN\|_{L^\infty}}{\sqrt{N}}\left(\iint\dx\dy|\wN(x-y)| \| a_x a_y \cR^{-1/2} \PhiN \|^2 \right)^\frac12 \times \nn\\
&\qquad \qquad \qquad \times\left(\iint \d x \d y |w_N(x-y)| \left\| a_y\cR^\frac12 \omega_3 \PhiNM \right\|^2\right)^\frac12.
\end{align}
We can simplify the right-hand side using  \eqref{eqn:quadratic:intro-2}  and Lemma \ref{lem:hartree}. Moreover, by \eqref{eq:omega_3-bound} and Lemma \ref{lem:sobolev}, we have 
\begin{equation}
|\omega_3|^2 \cN \cR \le \frac{C_\eps N^\eps}{M^2} \cN^2 \dG_1(1-\Delta)\le C_\eps N^\eps   \dG_1(1-\Delta)
\end{equation}
on $\cF^{\le M}$. Combining this with the kinetic bound in Lemma \ref{lem:nam}, we find that
\begin{align}
\iint \d x \d y |w_N(x-y)| \left\| a_y\cR^\frac12 \omega_3 \PhiNM \right\|^2 &= \|w_N\|_{L^1}  \left \langle  \PhiNM, |\omega_3|^2 \cN \cR  \PhiNM \right\rangle \nn\\
&\le C_\eps N^{2\eps}
\end{align}
for every constant $\eps>0$. Therefore, we deduce from \eqref{eq:E3-0} that
\begin{align}\label{eq:E3-final}
\mathcal{E}_3  \le \frac{C_{t,\eps} N^\eps}{\sqrt{N} }  .
\end{align}

\medskip
\noindent
{\bf Step 3.} Now we turn to the complicated error term $\cE_4$. A direct computation shows that 
\begin{equation}
\left[\dGt(|\wN|),\ad_x\ad_ya_y\right]=|\wN(x-y)|\ad_x\ad_ya_y+\int\dz|\wN(x-z)|\ad_x\ad_y\ad_z a_za_y\,,
\end{equation}
and with \eqref{eq:cR12-com-B} this yields
\begin{align}\label{eqn:commutator:R12-cubic}
\left[\cR^\frac12,\ad_x\ad_y a_y \right]&=\frac{1}{\pi}\int_0^\infty\ds\,\frac{ \sqrt{s}}{\cR+s} |\wN(x-y)|\ad_x\ad_y a_y  \frac{1}{\cR+s} \\ 
&+\frac{1}{\pi}\int_0^\infty\ds \int \d z \,\frac{ \sqrt{s}}{\cR+s}  |\wN(x-z)| \ad_x\ad_y\ad_z a_za_y \frac{1}{\cR+s}.\nn
\end{align}
Thus,  by the triangle inequality and the bound $\|\uN\|_{L^\infty}\le C_t$ from Lemma \ref{lem:hartree}, we can split 
\begin{align} \label{eq:cE4-0}
\mathcal{E}_4 &= \frac{1}{\sqrt{N} } \iint \dx\dy |\wN(x-y)| |\uN(x)| \Big|\Big\langle \PhiN, \cR^{-1/2}[\ad_x \ad_y a_y, \cR^{1/2}] \omega_3 \Phi_{N,M} \Big\rangle \Big|\nn\\
&\le C_t  (\cE_{4,1} + \cE_{4,2}) \,, 
\end{align}
where
\begin{align}
\cE_{4,1} &= \frac{1}{\sqrt{N}} \int_0^\infty \d s \sqrt{s}   \iint \dx\dy |\wN(x-y)|^2  \times\nn\\
&\qquad \qquad \qquad \qquad \qquad \times\Big|\Big\langle \PhiN,  \frac{\cR^{-1/2}}{\cR+s} \ad_x \ad_y a_y \frac{1}{\cR+s}\omega_3 \Phi_{N,M} \Big\rangle \Big|\\
\cE_{4,2} &= \frac{1}{\sqrt{N}} \int_0^\infty \d s \sqrt{s}   \iiint \dx\dy \d z |\wN(x-y)| |\wN(x-z)| \times\nn \\
&\qquad \qquad \qquad \qquad \qquad \times  \Big|\Big\langle \PhiN, \frac{\cR^{-1/2}}{\cR+s} \ad_x \ad_y \ad_z a_z a_y \frac{1}{\cR+s}\omega_3 \Phi_{N,M} \Big\rangle \Big|.
\end{align}

\medskip
\noindent
{\bf Estimate of $\cE_{4,1}$.}
By the Cauchy--Schwarz inequality we have
\begin{align}
\cE_{4,1} &\le \frac{1}{\sqrt{N}} \int_0^\infty \d s \sqrt{s}  \iint \dx\dy |\wN(x-y)|^2  \left\| a_x a_y \frac{\cR^{-1/2}}{\cR+s} \PhiN \right\| \left\| a_y \frac{1}{\cR+s}\omega_3 \Phi_{N,M} \right\|\nn\\
&\le \frac{1}{\sqrt{N}} \int_0^\infty \d s \sqrt{s}   \left(  \iint \dx\dy |\wN(x-y)|^{2-\eps}  \left\| a_x a_y \frac{\cR^{-1/2}}{\cR+s} \PhiN \right\|^2 \right)^{\frac 1 2} \times \nn\\
&\qquad \times \left(\iint \d x \d y |w_N(x-y)|^{2+\eps} \left\| a_y \frac{1}{\cR+s}\omega_3 \Phi_{N,M} \right\|^2 \right)^{\frac 1 2}. 
\end{align}
The right-hand side can be simplified using \eqref{eq:cE21-2} and the estimate 
\begin{align}
&\iint \d x \d y |w_N(x-y)|^{2+\eps} \left\| a_y \frac{1}{\cR+s}\omega_3 \Phi_{N,M} \right\|^2 \nn\\
&= \||w_N|^{2+\eps}\|_{L^1} \left\langle \Phi_{N,M}, \frac{\cN |\omega_3|^2}{(\cR+s)^2}  \Phi_{N,M}\right\rangle\le C_{t,\eps}N^{(1+\eps)2\beta} \frac{N^\eps}{M^2 (1+s)^2} \,,
\end{align}
which follows from \eqref{eq:omega_3-bound},  $\cR\ge 1$, and the kinetic bound in Lemma \ref{lem:nam}. Altogether, this gives
\begin{align} \label{eq:E41-final}
\cE_{4,1} & \le \frac{C_{t,\eps}}{\sqrt{N}}  \int_0^\infty \d s\sqrt{s}  \sqrt{ \frac{1}{(1+s)^{1+\eps}}  }\sqrt{ N^{(1+\eps)2\beta} \frac{N^\eps}{M^2 (1+s)^2} } \nn\\
&\le C_{t,\eps} \frac{N^{(1+\eps)\beta + \eps/2}}{\sqrt{N}M}.
\end{align}

\bigskip
\noindent
{\bf Estimate of $\cE_{4,2}$.} By the Cauchy--Schwarz inequality,
\begin{align}\label{eq:E42-0}
\cE_{4,2 } &= \frac{1}{\sqrt{N}} \int_0^\infty \d s \sqrt{s}   \iiint \dx\dy \d z |\wN(x-y)| |\wN(x-z)| \times \nn \\
&\qquad \qquad \qquad   \times  \Big|\Big\langle \PhiN, \frac{\cR^{-1/2} (\cN+2)^{-1/2}}{\cR+s} \ad_x \ad_y \ad_z a_z a_y \frac{(\cN+3)^{1/2}}{\cR+s}\omega_3 \Phi_{N,M} \Big\rangle \Big|\nn\\
&\le  \frac{1}{\sqrt{N}} \int_0^\infty \d s \sqrt{s}   \iiint \dx\dy \d z |\wN(x-y)| |\wN(x-z)| \times \nn\\
&\qquad \qquad \qquad   \times  \left\| a_x a_y a_z \frac{\cR^{-1/2}(\cN+2)^{-1/2}}{\cR+s} \PhiN \right\| \left\| a_y a_z \frac{(\cN+3)^{1/2}}{\cR+s}\omega_3 \Phi_{N,M} \right\|
\nn\\
&\le \frac{1}{\sqrt{N}} \int_0^\infty \d s \sqrt{s} \times  \nn\\
&\quad\times \left( \iiint \d x \d y \d z |\wN(x-y)|^{2-\eps} \left\| a_x a_y a_z \frac{\cR^{-1/2} (\cN+2)^{-1/2}}{\cR+s} \PhiN \right\|^2 \right)^{\frac 1 2} \nn\\
&\quad \times \left( \iiint \d x \d y \d z |\wN(x-y)|^{\eps} |\wN(x-z)|^2 \left\| a_y a_z \frac{(\cN+3)^{1/2}}{\cR+s}\omega_3 \Phi_{N,M} \right\|^2 \right)^{\frac 1 2}.
\end{align}
We can bound 
\begin{align} \label{eq:E42-1}
&\iiint \d x \d y \d z |\wN(x-y)|^{2-\eps} \left\| a_x a_y a_z \frac{\cR^{-1/2} (\cN+2)^{-1/2}}{\cR+s} \PhiN \right\|^2 \nn\\
&= \iint \d x \d y  |\wN(x-y)|^{2-\eps} \left\|  a_x a_y \frac{\cR^{-1/2}}{\cR+s} \PhiN \right\|^2 \le \frac{1}{(1+s)^{1+\eps}}  
\end{align}
as in \eqref{eq:cE21-2}. Since $w$ is bounded and compactly supported, we have the pointwise estimate
\begin{align}
|\wN(x-y)|^{\eps} |\wN(x-z)|^2
&=|\wN(x-y)|^{\eps} |\wN(x-z)|^2 \id_{\{|y-z|\le CN^{-\beta}\}}\nn\\
&\le C N^{4\beta} |\wN(x-y)|^{\eps} \id_{\{|y-z|\le CN^{-\beta}\}}\,. 
\end{align}
Moreover, note that the operators $\dG_2 \left( \id_{\{|y-z|\le CN^{-\beta}\}}\right)$, $\cR$, $\cN$ and $\omega_3$ all commute. Consequently, using $\cR\ge 1$ and \eqref{eq:omega_3-bound}, we can bound 
\begin{align}
&\iint  \d y \d z \id_{\{|y-z|\le CN^{-\beta}\}}  \left\| a_y a_z \frac{(\cN+3)^{1/2}}{\cR+s}\omega_3 \Phi_{N,M} \right\|^2\nn\\
&= \left\langle \Phi_{N,M} , \dG_2 \left( \id_{\{|y-z|\le CN^{-\beta}\}}  \right) \frac{\cN+3}{(\cR+s)^2}|\omega_3|^2 \Phi_{N,M}  \right\rangle \nn\\
&\le \frac{C}{M(1+s)^2} \left\langle \Phi_{N,M}, \dG_2 \left( \id_{\{|y-z|\le CN^{-\beta}\}}  \right)  \Phi_{N,M} \right\rangle. 
\end{align}
Using Lemma \ref{lem:sobolev} with $s=2\beta/(2\beta-\varepsilon)$, we obtain
\begin{equation}
\dG_2 \left( \id_{\{|y-z|\le CN^{-\beta}\}}  \right) \le C_\eps N^\eps N^{-2\beta} \cN \dG_1(1-\Delta) 
\end{equation}
for every $\eps>0$. Therefore, together with Lemma \ref{lem:nam}, we deduce that
\begin{align} \label{eq:E42-2}
&\iiint \d x \d y \d z |\wN(x-y)|^{\eps} |\wN(x-z)|^2 \left\| a_y a_z \frac{(\cN+3)^{1/2}}{\cR+s}\omega_3 \Phi_{N,M} \right\|^2\nn\\
&\le \frac{CN^{4\beta}}{M(1+s)^2}  \| |w_N|^{\eps} \|_{L^1} \left \langle \Phi_{N,M}, \dG_2 \left( \id_{\{|y-z|\le CN^{-\beta}\}}  \right)  \Phi_{N,M}\right\rangle\nn\\
&\le \frac{C_{t,\eps} N^{(2\beta+2) \eps}}{(1+s)^2}. 
\end{align}
Inserting \eqref{eq:E42-1} and \eqref{eq:E42-2} in \eqref{eq:E42-0} we find that
\begin{align}\label{eq:E42-final}
\cE_{4,2 } = \frac{C_{t,\eps}}{\sqrt{N}} \int_0^\infty \d s \sqrt{s}  \sqrt{\frac{1}{(1+s)^{1+\eps}} } \sqrt{ \frac{N^{(2\beta+2) \eps}}{(1+s)^2} } \le  \frac{C_{t,\eps}N^{(\beta+1)\eps}}{\sqrt{N}} 
\end{align}
for every constant $\eps>0$. From \eqref{eq:E41-final} and \eqref{eq:E42-final} we get
\begin{align}\label{eq:E4-final}
\cE_4\le C_{t,\eps} \left( \frac{N^{\beta}} {\sqrt{N}M} + \frac{1} {\sqrt{N}} \right) N^\eps .
\end{align}

\noindent
{\bf Conclusion of \eqref{eq:B-cubic-1}}: Given the decomposition \eqref{eq:G3-dec}, the desired bound \eqref{eq:B-cubic-1} follows immediately from  \eqref{eq:E3-final} and \eqref{eq:E4-final}. 

\bigskip
\noindent
{\bf Step 4.} It remains to prove \eqref{eq:B-cubic-2}. Similarly to \eqref{eq:B-GB3-0}, we can write
\begin{align} \label{eq:B-GB3*-0}
\lr{\PhiN,\left[\bG_3^*,\fM^2\right]\PhiNM} = \frac{1}{\sqrt{N}} \iint \dx\dy \wN(x-y) \overline{\uN(x)}\Big\langle \PhiN, \ad_y a_x a_y  \tilde{\omega}_3 \Phi_{N,M} \Big\rangle
\end{align}
with 
\begin{equation}
\tilde \omega_3=\sqrt{1-\frac{\Number}{N-1}} \left( f^2\Big(\frac{\Number-1}{M}\Big) - f^2\Big(\frac{\Number}{M}\Big) \right) 
\end{equation}
as operator on $\FNp$, which satisfies 
\begin{equation}\label{eq:omega_3-t-bound}
|\tilde \omega_3|\leq \frac{C}{M} \id^{\le M}.
\end{equation}
By the Cauchy--Schwarz inequality,
\begin{align} \label{eq:B-GB3*-final}
&\left|\lr{\PhiN,\left[\bG_3^*,\fM^2\right]\PhiNM} \right| \nn\\
&=   \frac{1}{\sqrt{N}} \left| \iint \dx\dy \wN(x-y) \overline{\uN(x)}\Big\langle (\cN+1)^{-1/2}\PhiN, \ad_y a_x a_y  \cN^{\frac12}\tilde{\omega}_3 \Phi_{N,M} \Big\rangle \right|\nn\\
&\le \frac{\|\uN\|_{L^\infty}}{\sqrt{N}} \iint \dx\dy |\wN(x-y)| \| a_y (\cN+1)^{-1/2} \PhiN \|  \| a_x a_y \cN^{1/2}  \tilde{\omega}_3 \Phi_{N,M} \|\nn\\
&\le \frac{C_t}{\sqrt{N}} \left(  \iint \dx\dy |\wN(x-y)| \| a_y (\cN+1)^{-1/2} \PhiN \| ^2 \right)^\frac 1 2  \nn\\
&\qquad \qquad \qquad \times \left(  \iint \dx\dy |\wN(x-y)| \| a_x a_y \cN^{1/2}  \tilde{\omega}_3 \Phi_{N,M}\|^2 \right)^\frac 1 2 \nn\\
&= \frac{C_t}{\sqrt{N}}  \left\langle \PhiN, \|w_N\|_{L^1} \cN (\cN+1)^{-1} \PhiN \right\rangle^\frac 1 2   \left\langle \Phi_{N,M}, \dG_2(|w_N|) \cN |\tilde \omega_3|^2 \Phi_{N,M} \right\rangle^\frac 1 2  \nn\\
&\le \frac{C_{t,\eps}N^\eps}{\sqrt{N}} \|\PhiN\| \left\langle \Phi_{N,M},  \dG_1(1-\Delta) \cN^2 \1^{\le M} M^{-2} \Phi_{N,M} \right\rangle^\frac 1 2 \le \frac{C_{t,\eps}}{\sqrt{N}}  N^{2\eps}\,, 
\end{align}
where we used Lemma \ref{lem:sobolev} and the kinetic bound in Lemma  \ref{lem:nam}. This concludes the proof of \eqref{eq:B-cubic-2} and thus of Lemma \ref{lem:cubic}.
\end{proof}

\subsection{Conclusion of Proposition \ref{prop:B}} First, inserting the bounds from Lemmas \ref{lem:linear}, \ref{lem:quadratic} and \ref{lem:cubic} in \eqref{eq:lemB-dec}, and using $M\le N$ to simplify some error terms, we find that the desired propagation bound \eqref{eq:probB-simplified} holds true, namely that
$$
\left|\frac{\d}{\d t} \lr{\PhiN(t),f_M^2\PhiNM(t)}  \right|\le C_{t,\varepsilon}N^\varepsilon\left(\frac{1}{\sqrt{M}}+\frac{N^{\beta}}{M^{3/2}}\right)\,.  
$$
Now we are ready to give 

\begin{proof}[Proof of Proposition \ref{prop:B}] Define 
\begin{equation}
\cB(t):=1-\Re\lr{\PhiN(t),\fM^2\PhiNM(t)}.
\end{equation}
By Assumption \eqref{eq:initial-Phi0} and by definition \eqref{eq:fM} of $\fM$, we obtain
\begin{equation}
\cB(0) = \lr{\Phi_0, (1-\fM^2) \Phi_0}    \leq \lr{\Phi_0,\id^{>\frac{M}{2}}\Phi_0}\leq\frac{2}{M}\lr{\Phi_0,\Number\Phi_0}\leq\frac{C}{M}\,.
\end{equation}
Combining this with  \eqref{eq:probB-simplified}, we can therefore bound 
\begin{align} \label{eq:cB-final-bound}
\cB(t) \le C_{t,\eps} \left( \frac{1}{\sqrt{M}} + \frac{N^\beta}{M^{3/2}} \right) 
\end{align}
for all $t\in [0,\Tmax)$ and $\eps>0$ by Gronwall's lemma.

To conclude Proposition \ref{prop:B}, we prove that 
\begin{equation} \label{eq:B}
\norm{\PhiN(t)-\PhiNM(t)}^2\leq 4\cB(t)\,.
\end{equation}
Let us drop the time dependence from the notation for simplicity and write 
\begin{align} 
\|\PhiN-\PhiNM\|^2 &= \| \Phi_N \|^2 + \|\PhiNM\|^2 - 2\Re \langle \Phi_N, \Phi_{N,M} \rangle \nn\\
&\le 2  - 2\Re \langle \Phi_N, f_M^2 \Phi_{N,M} \rangle - 2\Re \langle \Phi_N(t), g_M^2 \Phi_{N,M} \rangle.
\end{align}
Here we denoted $g_M^2=1-f_M^2$ and  used that $\|\Phi_N\|\le 1$, $\|\PhiNM\|\le 1$. Moreover, by the Cauchy--Schwarz inequality, 
\begin{align} 
2\Big| \langle \Phi_N, g_M^2 \Phi_{N,M}\rangle\Big| &\le \|g_M \Phi_{N}\|^2 + \|g_M \Phi_{N,M}\|^2 \nn\\
&\le 2 - \|f_M \Phi_{N}\|^2 -  \|f_M \Phi_{N,M}\|^2 \le 2 - 2 \Big| \langle \Phi_{N}, f_M^2 \Phi_{N,M}\rangle \Big|. 
\end{align}
Thus,  \eqref{eq:B} follows immediately. The proof of Proposition \ref{prop:B} is complete.  
\end{proof}

\section{Conclusion of the main theorems}\label{sec:conclusion}

\subsection{Proof of Theorem \ref{thm:2}} Let $M=N^{1-\delta}$ with $\delta\in (0,1)$. Recall that $\Phi_N(t)$ and $\Phi_{N,M}(t)$ are defined in \eqref{eq:PhiN} and \eqref{eq:SE:truncated}, respectively. Since $U_N: \fH^N \to \FNp(t)$ is a unitary transformation, the desired norm approximation \eqref{eq:norm-PhiNt} is equivalent to 
\begin{equation} \label{eq:norm-PhiNt-1}
\| \Phi_N(t) - \Phi(t) \|^2 \le C_{t}  N^{-\alpha_2}. 
\end{equation}
By Lemma \ref{lem:nam} and  Proposition \ref{prop:B}, we can bound
\begin{align}
\|\PhiN(t)-\Phi(t) \|^2 &\le 2 \|\PhiN(t)-\Phi_{N,M}(t) \|^2 + 2\|\Phi_{N,M}(t)-\Phi(t) \|^2 \nn  \\
&\le C_{t,\eps} N^\eps \left(\frac{1}{\sqrt{M}}+\frac{N^\beta}{M^{3/2}} + \sqrt{\frac{M}{N}}  \right)\nn\\
&=C_{t,\eps} N^\eps \left(N^{\frac{\delta-1}{2}}+ N^{\frac{3\delta}{2}+\beta-\frac32}+N^{-\frac{\delta}{2}}\right)
\end{align}
for all $t\in[0,\Tmax)$ and $\eps>0$. Here we have put back $M=N^{1-\delta}$ at the end. The optimal choice for $\delta$ is
\begin{equation}
\delta=\begin{cases}\frac{3-2\beta}{4} & \text{ if } \beta\ge \frac12\\
\frac 1 2 &\text{ if }\beta \le \frac12\end{cases}
\end{equation}
which implies \eqref{eq:norm-PhiNt-1} with every $0<\alpha_2<\min(1/8, (3-2\beta)/16)$ .\qed

\subsection{Proof of Theorem \ref{thm:1}} The implication of the convergence of density matrices from the norm convergence is well-known, see e.g. \cite[Corollary 2]{LewNamSch-15}.  Here we recall a quick derivation for the reader's convenience.  We will again drop the time dependence from the notation. Let $ q = 1- p=  1- |\uN\rangle \langle \uN|$ as in \eqref{def:p:q}.  By  using the rules \eqref{eqn:substitution:rules} (see also Remark \ref{rmk:UN-PhiN}), Theorem \ref{thm:2} and Lemma \ref{lem:Bog}, it follows that
$$
N\Tr(q \gamma_{\Psi_N}^{(1)} q)= \| \sqrt{\cN}  \Phi_N\|^2 \le 2\|  \sqrt{\cN}  \1^{\le N} ( \Phi_N - \Phi) \|^2 +2\|  \sqrt{\cN}    \Phi \|^2 
\le C_{t,\eps} \Big( N^{1-{2\alpha_2}} + N^\eps \Big).
$$
Then by the triangle and  Cauchy--Schwarz inequalities, we conclude that 
\begin{align*} 
\Tr |\gamma_{\Psi_1}^{(1)}-|\varphi\rangle \langle \varphi|| &\le \Tr| p- |\varphi\rangle \langle \varphi| |  + \Tr | p(\gamma_{\Psi_1}^{(1)}-1)p|  + \Tr |q \gamma_{\Psi_1}^{(1)} q | + 2 \Tr |p \gamma_{\Psi_1}^{(1)} q| \\
&\le 2 \|u_N -\varphi\|_{L^2} +2\Tr (q\gamma_{\Psi_1}^{(1)} q) + 2 \sqrt{ {\Tr | q\gamma_{\Psi_1}^{(1)} q|} } \sqrt{\Tr | p\gamma_{\Psi_1}^{(1)} p|} \nn\\
&\le C_{t,\eps} (N^{-\beta}+ N^{-\alpha_2} + N^{\eps}). 
\end{align*}
Here we used  $\Tr (p)= \Tr \, \gamma_{\Psi_N}^{(1)}=1$ and Lemma \ref{lem:hartree}. Thus \eqref{eq:1pdm-cv-phi} holds for $\alpha_1=\min(\beta,\alpha_2)$. 
The proof of Theorem \ref{thm:1} is complete. \qed

\bigskip
\noindent
{\bf Acknowledgements.} We would like to thank Kihyun Kim for helpful discussions.
L.~Bo{\ss}mann was supported by the Deutsche Forschungsgemeinschaft (DFG, German Research Foundation) via the Munich Center of Quantum Science and Technology (Germany's
Excellence Strategy EXC-2111-390814868). C.~Dietze and P.T.\ Nam were supported by the DFG project ``Mathematics of many-body quantum systems''  (project No.\ 426365943). C.~Dietze also  acknowledges the partial support from the Jean-Paul Gimon Fund and from the Erasmus+ programme.

\bibliographystyle{siam}
\bibliography{biblio-Nam-Lea.bib}

\end{document}